%% file: main-arxiv.tex
\documentclass[a4paper,UKenglish,cleveref, autoref, thm-restate]{lipics-v2021}

\hideLIPIcs

\usepackage{amsmath, amsfonts, amssymb}
\usepackage{graphics}
\usepackage{latexsym}
\usepackage{tikz}

\usepackage{float}

\usepackage[linesnumbered,algoruled,boxed,lined,vlined]{algorithm2e}

\newtheorem{invariant}{Invariant}

\newcommand{\lpar}{\left (}
\newcommand{\rpar}{\right )}

\newcommand{\lbrac}{\left \lbrace}
\newcommand{\rbrac}{\right \rbrace}

\newcommand{\graham}[3]{\mbox{\ensuremath{#1\mid#2\mid#3}}}
\newcommand{\opt}{\mathit{OPT}}
\newcommand{\alg}{\mathit{ALG}}

\newcommand{\jobset}{\ensuremath{\mathcal{J}}}
\newcommand{\objsum}{\ensuremath{\sum_j C_j}}
\newcommand{\compl}{\ensuremath{C_j}}

\usepackage{pdfpages}

\newcommand{\proc}[1][j]{\ensuremath{p_{#1}}}
\newcommand{\test}[1][j]{\ensuremath{t_{#1}}}
\newcommand{\total}[1][j]{\ensuremath{\sigma_{#1}}}
\newcommand{\delay}[2][i, j]{\ensuremath{d_{#1}^{#2}}}
\newcommand{\pdelay}[2][i, j]{\ensuremath{D^{#2}{(#1)}}}
\newcommand{\nodes}{\ensuremath{V}}

\newcommand{\arcs}{\ensuremath{A}}

\newcommand{\red}{\ensuremath{R}}

\graphicspath{ {./img/} }

\title{Scheduling with Obligatory Tests}

\author{Konstantinos Dogeas}{Department of Computer Science, Durham University, UK}{konstantinos.dogeas@durham.ac.uk}{https://orcid.org/0009-0001-1528-3221}{Funded by EPSRC grant EP/S033483/2.}

\author{Thomas Erlebach}{Department of Computer Science, Durham University, UK}{thomas.erlebach@durham.ac.uk}{https://orcid.org/0000-0002-4470-5868}{Supported by EPSRC grants EP/S033483/2 and EP/T01461X/1.}

\author{Ya-Chun Liang}{Data Science Institute, Columbia University, USA}{yl5608@columbia.edu}{}{}

\authorrunning{K. Dogeas et al.}

\Copyright{Konstantinos Dogeas and Thomas Erlebach and Ya-Chun Liang}

\ccsdesc[500]{Theory of computation~Scheduling algorithms}

\keywords{Competitive ratio, Online algorithm, Scheduling with testing, Sum of completion times}

\relatedversion{An extended abstract of this paper appears in the proceedings of the 32nd Annual European Symposium on Algorithms (ESA 2024).}

\nolinenumbers

\begin{document}
\maketitle

\begin{abstract}
Motivated by settings such as medical treatments or aircraft
maintenance, we consider a scheduling problem with
jobs that consist of two operations, a test and a processing part.
The time required to execute the test is known in advance while the time required
to execute the processing part becomes known only upon completion of the test.
We use competitive analysis to study algorithms for
minimizing the sum of completion times for $n$ given jobs on a single machine.
As our main result, we prove using a novel analysis technique
that the natural $1$-SORT algorithm has competitive ratio
at most $1.861$. For the special case of uniform test times,
we show that a simple threshold-based algorithm has
competitive ratio at most $1.585$. We also prove a lower bound
that shows that no deterministic algorithm can be better
than $\sqrt{2}$-competitive even in the case of uniform test
times.
\end{abstract}

\section{Introduction} \label{sec:intro}
Settings where the processing time of a job is initially uncertain but
can be determined by executing a test have received increasing
attention in recent years. 
Levi et al.~\cite{LeviMS19} considered a setting where the weight
and processing time of a job follow a known joint probability distribution,
a job can be tested to reveal its weight and processing time,
and the goal is to find a scheduling policy that minimizes the expectation
of the weighted sum of completion times.
D\"urr et al.~\cite{DurrEMM20} introduced an adversarial setting 
of scheduling with testing where
each job $j$ is given with an upper bound $u_j$ on its processing time.
The scheduler can either execute the job untested (with processing
time~$u_j$), or test it first to reveal its actual processing time
$p_j\le u_j$ and then execute it with processing time~$p_j$.
They studied the setting of uniform test times and gave
competitive algorithms for minimizing the sum of completion
times and for minimizing the makespan on a single machine.
Subsequent work considered this
adversarial model with arbitrary test times for minimizing the sum
of completion times on a single machine~\cite{AlbersE20,LiuLWZ23}
or on multiple machines~\cite{GongCH23},
for makespan minimization on parallel machines~\cite{AlbersE21,GongL21,GongGLM22},
and for minimizing energy or maximum speed in scheduling with speed
scaling~\cite{BampisDKLP21}.

In all these studies, it is \emph{optional} for the scheduler whether to test a job
or not. In many application settings, however, it is natural
to assume that a test \emph{must} be executed for each job before the job can
be executed. For example, for a repair job, it is necessary to first
diagnose the fault (this corresponds to a test) before the repair can
be carried out, and the result of the fault diagnosis yields information
about how long the repair job will take. For a maintenance job (for example,
aircraft maintenance~\cite{LeviMS19}), it is
necessary to determine the maintenance needs (this corresponds to a test)
before the maintenance can be carried out. In a medical emergency department,
patients need to be diagnosed (i.e., `tested') before they can be treated. Therefore,
we propose to study scheduling with testing in a setting with \emph{obligatory tests}. Initially,
each job $j$ is given with a test time~$t_j$, and nothing is known
about its processing time. Testing the job takes time~$t_j$ and
reveals the processing time $p_j$ of the job. The processing part
of the job can then
be scheduled any time after the completion of the test and takes
time $p_j$ to be completed. We study algorithms for minimizing the
sum of the completion times on a single machine and evaluate the
performance of our algorithms using competitive analysis. Note that
in our setting the offline optimum must also test every job.

We consider both the setting with arbitrary test times and the setting
with uniform test times where we assume w.l.o.g.~that $t_j=1$ for all jobs~$j$.
The latter setting
is motivated by applications where the test operation takes
the same time for every job; for example, in a medical setting, every
patient may have to undergo the same test procedure to be diagnosed.

For minimizing the sum of completion times on a single machine
in our setting with obligatory tests, obtaining a $2$-competitive
algorithm is straightforward: Treating each job (test plus processing part)
as a single entity with unknown processing time and applying the
Round Robin (RR) algorithm (which executes all unfinished jobs
simultaneously at the same rate) gives a $2$-competitive preemptive
schedule~\cite{MotwaniPT94},  and in our setting this algorithm can
be made non-preemptive without any increase in job completion times:
At any time, among all tests or processing parts currently available
for execution, it is known which of them will complete first in the
preemptive schedule, and hence that test or processing part can be
chosen to be executed non-preemptively first (the same observation
has been made previously for the setting with optional
tests~\cite{DurrEMM20,LiuLWZ23,GongCH23}). Our aim is therefore to
design algorithms that are better than $2$-competitive.

\subsection{Our contributions}
For the setting with arbitrary test times, we consider the
algorithm $1$-SORT, which is a natural adaptation of the
$(\alpha,\beta)$-SORT algorithm proposed by Albers and Eckl~\cite{AlbersE20}
to the setting with obligatory tests. Using a novel analysis
technique that we consider our main contribution, we show
that the competitive ratio of $1$-SORT is at most $1.861$.
In our analysis, we consider a complete graph
on the jobs, where each edge is associated with the delay that
the two jobs connected by the edge create for each other.
The sum of the delays associated with the edges
and the job processing times is then equal to the sum of
completion times of the schedule.
The graph can contain edges where the associated delay in the schedule
computed by the algorithm is arbitrarily close to twice the delay in the
optimal schedule, and therefore a straightforward analysis
would only yield a competitive ratio of~$2$. We show that
for edges with delay ratio close to~$2$ there are always
sufficiently many other edges whose delay ratio is much
smaller than~$2$, so that overall the ratio of the objective
values of the algorithm and the offline optimum is bounded
by a value smaller than~$2$.

For the setting with unit test times, we consider
an adaptation of the \textsc{Threshold} algorithm by D\"urr et al.~\cite{DurrEMM20}
to the setting with obligatory tests: When the test
of a job reveals a processing time smaller than
a threshold $y$, the algorithm executes the processing part of the job
immediately; otherwise, the execution of the processing part is deferred
to the end of the schedule, where all the processing parts that have
been deferred are executed in SPT (shortest processing time) order.
We show that the algorithm is $1.585$-competitive (and this analysis is tight
for the algorithm). We also give a lower bound showing that
no deterministic algorithm can be better than $\sqrt{2}$-competitive.

\subsection{Related Work} \label{subsec:related_work}
For the classical offline scheduling problem (without tests) of minimizing the sum of completion times
on a single machine, denoted by $1\mid\,\mid\sum C_j$,
it is known that always executing first the job with the shortest processing
time (SPT) among all unscheduled jobs gives the optimal schedule (a generalisation
to the weighted sum of completion times was proven by Smith~\cite{Smith56}).
For the setting with unknown processing times (i.e., the scheduler
does not know the processing time of a job until the job completes),
Motwani et al.~\cite{MotwaniPT94} 
showed that the Round Robin (RR) algorithm, a preemptive algorithm that schedules 
all unfinished jobs simultaneously,
is $\lpar 2 - \frac{2}{n+1} \rpar$-competitive, where $n$ is the number of
jobs, and that this is best possible.

As mentioned earlier, D\"urr et al.~\cite{DurrEMM20} introduced the adversarial model for
scheduling with testing in a setting with optional tests: For each job $j$ its test
time $t_j$ and an upper bound $u_j$ on its processing time are given. The algorithm
can either execute the job untested with processing time $u_j$ or test it first.
The test takes time~$t_j$ and reveals the actual processing time $p_j$, which satisfies
$0\le p_j\le u_j$. The job can then be executed at any time after the test and takes
time~$p_j$. They considered only the case of uniform test times ($t_j=1$ for all jobs~$j$)
and provided a $2$-competitive deterministic algorithm and a $1.7453$-competitive
randomized algorithm for minimizing the sum of completion times on a single machine. Their deterministic $2$-competitive algorithm is the algorithm \textsc{Threshold} that tests all jobs with $u_j\ge 2$ and executes the processing part of a job $j$ immediately after its test if $p_j\le 2$ and otherwise defers the job to the end of the schedule (where the processing parts of all unfinished jobs are executed in SPT order).
They also gave lower bounds of $1.8546$ and $1.6257$ for deterministic and
randomized algorithms, respectively. Albers and Eckl~\cite{AlbersE20}
considered the problem with arbitrary test times and gave a deterministic
$4$-competitive algorithm, a $3.3794$-competitive randomized algorithm,
and a preemptive deterministic algorithm with competitive ratio $2\phi\approx 3.2361$,
where $\phi=(1+\sqrt{5})/2$ is the golden ratio.
Their preemptive deterministic algorithm can be made non-deterministic as outlined
above, thus giving a $2\phi$-competitive deterministic algorithm.
The algorithm for which they showed competitive ratio~$4$ is called $(\alpha,\beta)$-SORT:
It tests a job $j$ if $u_j\ge \alpha t_j$ and, at any time, executes the test or processing part
of smallest priority, where the priority of the test of a job $j$ is taken to be $\beta t_j$ and the priority of the processing part of a tested job $j$ is taken to be~$p_j$. In their
analysis, choosing $\alpha=\beta=1$ optimizes the resulting ratio, giving the bound
of~$4$.
Liu et al.~\cite{LiuLWZ23} showed that a more careful analysis of $(\alpha,\beta)$-SORT
yields that the algorithm achieves ratio $1+\sqrt{2}\approx 2.414$ for $\alpha=\beta=\sqrt{2}$.
They also gave improved algorithms with deterministic
competitive ratio $2.316513$ and randomized competitive ratio $2.152271$.
Gong et al.~\cite{GongCH23} considered the problem of minimizing the sum of completion
times in the setting with optional tests on multiple machines. Among other results,
they presented a $3.2361$-competitive algorithm for arbitrary test times and
an algorithm with competitive ratio approaching $2.9271$ for large $m$
for uniform test times.

For the problem of minimizing the makespan in scheduling with optional tests
on a single machine, D\"urr et al.~\cite{DurrEMM20}
gave a deterministic $\phi$-competitive algorithm and a randomized $4/3$-competitive algorithm
for uniform test times,
and showed that both bounds are best possible. Albers and Eckl~\cite{AlbersE20} showed that
the same bounds hold for the case of arbitrary test times.
Albers and Eckl~\cite{AlbersE21} then considered the case of $m$ parallel machines
and gave algorithms with competitive ratio~$2$ for the case where the processing part of a job
can be executed at any time after the completion of the test, possibly even on a different
machine. For the setting where the processing part of a job must be executed
immediately after its test on the same machine, they presented algorithms with
ratios approaching $3.016$ for arbitrary test times and $3$
for uniform test times for large~$m$. The latter ratios were improved to
$2.9513$ and $2.8081$, respectively, by Gong and Lin~\cite{GongL21},
and to $2.8681$ and $2.5276$, respectively, by Gong et al.~\cite{GongGLM22}.

\subsection{Outline of the paper}
We give a formal problem definition and discuss preliminaries in Section~\ref{sec:prelim}.
Our algorithm for arbitrary test times and its analysis are presented in Section~\ref{sec:arbitrary}.
The threshold-based algorithm and its analysis as well as the lower bound for uniform test times
are given in Section~\ref{sec:uniform}.
Conclusions are presented in Section~\ref{sec:conc}.

\section{Problem Definition and Preliminaries}
\label{sec:prelim}%
\subparagraph{Problem definition.}
We are given a job set $\jobset = \lbrac 1, 2, \dots, n \rbrac$ that must be scheduled on a single 
machine.
Each job~$j \in \jobset$ has an \emph{unknown} processing time $\proc \ge 0$ and a \emph{known} 
test time $\test \ge 0$, where $\proc$ and $\test$ are non-negative real numbers.
We denote the \emph{total size} (or just \emph{size}) of job $j$ by $\total = \test + \proc$.
Furthermore, we denote the maximum of the test time and the processing time
of job $j$ by $m_j = \max\{t_j,p_j\}$.
Testing job~$j$ takes time $\test$ and reveals its processing time~$p_j$.
Once job~$j$ has been tested, its processing part can be executed and takes
time~$p_j$. The completion time $C_j$ of a job is the point in time when
its processing part finishes.
We consider the setting with \emph{obligatory tests} where every job must be tested 
before the processing part of the job can be executed.
Note that the test of every job must be executed both by the algorithm and by the optimal solution.
The machine can execute at any time only one test or one processing part of a job.
The tests and processing parts must be scheduled non-preemptively, but the
processing part of a job does not have to be started immediately after its test.
As is common in the literature on scheduling with testing for minimizing the
sum of completion times~\cite{DurrEMM20,AlbersE20,GongCH23}, we refer to this setting
as \emph{non-preemptive} but note that it has been called \emph{test-preemptive}
in the context of makespan minimization~\cite{AlbersE21,GongL21,GongGLM22}.
The objective is to minimize the sum of completion times $\sum_{j\in\jobset} C_j$. 

In the setting of \emph{uniform test times}, we assume that $t_j=1$ for all $j\in\jobset$.
In the setting of \emph{arbitrary test times}, the test time of each job $j$ is an arbitrary
real number~$t_j\ge 0$.
Using Graham's notation for describing scheduling problems, these two variations can
be denoted by
$\graham{1}{\test = 1}{\objsum}$ and $\graham{1}{\test}{\objsum}$, respectively.

\subparagraph{The objective function.}
For the purpose of analyzing the competitive ratio of algorithms,
it will be useful to consider different ways of expressing the
objective function.
For two different jobs $k$ and $j$ in the schedule produced
by the algorithm under consideration, we use $d_{k,j}$ to denote
the amount of time that the test and/or processing part of job
$k$ get executed before the completion of job~$j$. The completion
time of job $j$ can then be written as
$$C_j=\sigma_j + \sum_{\substack{k \in \jobset \\ k \neq j}} \delay[k,j]{} \,.$$
For any pair of different jobs $j$ and~$k$, we use $D(j,k)=d_{j,k}+d_{k,j}$
to denote the delay that job $j$ causes for job $k$ plus the delay that job $k$
causes for job~$j$. We then have:
\begin{equation} \label{eq:obj_alg}
	\sum_{j \in \jobset} \compl  = \sum_{j \in \jobset} \lpar \sigma_j + \sum_{\substack{k \in \jobset \\ k \neq 
	j}} \delay[k, j]{} \rpar = \sum_{j \in \jobset} \sigma_j + \sum_{\substack{j, k \in \jobset \\ j < k}} 
	\pdelay[j, k]{}
\end{equation}

For the case of uniform test times, we will sometimes use that
a schedule of $n$ jobs with $p_1\ge p_2\ge\cdots\ge p_n$
that executes the jobs in SPT order, with the execution of the test of each job
immediately followed by the execution of its processing part,
has sum of completion times
\begin{equation}
\label{eq:uniformsum}
\sum_{j=1}^n C_j = \sum_{j=1}^n j(1+p_j) = \lpar \sum_{j=1}^n j\rpar + \lpar \sum_{j=1}^n p_j\rpar
=\frac{n(n+1)}{2} + \lpar \sum_{j=1}^n p_j\rpar
\,.
\end{equation}

\subparagraph{The optimal schedule.}
An optimal offline schedule views each job as a single operation
that takes total time $\sigma_j$ to be executed
and schedules the jobs in SPT order with respect
to those times. We use $d^*_{j,k}$ and $D^*(j,k)$ to denote the values
corresponding to $d_{j,k}$ and $D(j,k)$ in the optimal schedule. For
jobs $k$ and $j$ with $\sigma_k<\sigma_j$, we have
$d^*_{k,j}=\sigma_k$, $d^*_{j,k}=0$ and $D^*(j,k)=\sigma_k$.
In general, $D^*(j,j')=\min\{\sigma_j,\sigma_{j'}\}$ for any
pair of jobs $j$ and~$j'$.
For the sum of completion times $\opt$ in the optimal schedule, we have:
\begin{align} \label{eq:obj_opt}
	\opt =
	\sum_{j \in \jobset} \sigma_j + \sum_{\substack{j, k \in \jobset \\ j < k}} D^*(j,k)
	=
	\sum_{j \in \jobset} \sigma_j + \sum_{\substack{j, k \in \jobset \\ j < k}}
	\min \{ \sigma_j,\sigma_k \}
\end{align}

\subparagraph{Competitive ratio.}
For an algorithm under consideration, we
use $\alg$ to denote the sum of completion times in the schedule
produced by the algorithm for a given instance. By $\opt$ we denote the sum of completion
times in the optimal offline schedule for that instance. We say that the algorithm
is \emph{$\rho$-competitive} (or has \emph{competitive ratio} at most $\rho$)
if $\alg/\opt \le \rho$ holds for all instances of the problem.
\section{Arbitrary Test Times} \label{sec:arbitrary}
In this section, we consider the problem $1\mid t_j \mid \sum C_j$ where
the test times can be arbitrary non-negative real numbers.
We refer to the test and the processing part of a job~$j$ as \emph{operations}
and denote the test operation by~$\tau_j$ and the processing operation
by~$\pi_j$.
In Section~\ref{subsec:arbitrary:algorithm}, we present the algorithm
$\beta$-SORT.
In Section~\ref{subsec:arbitrary:upper}, we prove an upper bound
of $1.861$ on the competitive ratio of $\beta$-SORT with $\beta=1$.
In Section~\ref{subsec:arbitrary:lower}, we present input examples
showing that the competitive ratio of $\beta$-SORT is at least
$1.618$ for $\beta=1$ and at least some larger value for all other
values of~$\beta$.
\subsection{Algorithm $\beta$-SORT} \label{subsec:arbitrary:algorithm}
\begin{algorithm}[tbp]
	\caption{$\beta$-SORT \label{alg:betaSort}}
	$\mathcal{R} = \emptyset$; // empty priority queue\\
	\For{$j\in\jobset$}{
	 insert the test operation $\tau_j$ with priority $\beta\times t_j$ into $\mathcal{R}$
	}
	\While{$\mathcal{R} \neq \emptyset$}{
		$o = \mathcal{R}.\mathrm{deleteMin()}$\;
		execute $o$\;
		\If{$o$ was the test operation $\tau_j$ of a job $j$}{
		    insert the processing operation $\pi_j$ with priority $p_j$ into $\mathcal{R}$
			}
		}
\end{algorithm}%
For the problem variant with optional tests,
Albers and Eckl~\cite{AlbersE20} proposed the algorithm
$(\alpha,\beta)$-SORT that tests a job $j$ if
$u_j\ge \alpha t_j$ and always schedules the shortest
available operation, but uses
$\beta\times t_j$ instead of $t_j$ when comparing
the test time of job~$j$ with the processing
time of another job that has already been tested.
They showed that the algorithm is $4$-competitive
with $\alpha=\beta=1$.
We adapt their algorithm to our setting
with obligatory tests. The parameter $\alpha$ is not
relevant in our setting as every job must be tested,
so we refer to the resulting algorithm as $\beta$-SORT
(see Algorithm~\ref{alg:betaSort}).
The algorithm maintains a priority queue $\mathcal{R}$ of available test and
processing operations (i.e., the test operations of jobs that
have not yet been tested and the processing parts of jobs
that have already been tested). The priority of a test
operation $\tau_j$ is $\beta \times t_j$ and the priority
of a processing operation $\pi_j$ is~$p_j$.
The algorithm always schedules next the operation with minimum
priority in $\mathcal{R}$ (returned and removed from $\mathcal{R}$
by the call to $\mathcal{R}.\mathrm{deleteMin}()$) and, if that operation
was a test, inserts the corresponding processing operation into~$\mathcal{R}$.

\subsection{Upper bound on the competitive ratio of $1$-SORT} \label{subsec:arbitrary:upper}
By adapting the analysis by Albers and Eckl~\cite{AlbersE20} in
a straightforward way, one gets that $\beta$-SORT is 
$\lpar 1 + \max \lbrac 1 + \frac{1}{\beta}, 1 + \beta \rbrac \rpar$-competitive.
This bound is minimized for $\beta=1$, showing that the competitive ratio
of $1$-SORT is at most~$3$.
We fix $\beta=1$ and
prove the substantially better bound of $1.861$ on
the competitive ratio of $1$-SORT. We do not believe that
$\beta$-SORT with a value of $\beta$ different from~$1$ has
a better competitive ratio than that obtained with $\beta=1$ in
our setting; adapting our analysis to values of $\beta$ different from~$1$,
we found that the resulting bound on the competitive ratio became
larger.

\subparagraph{Intuitive overview of analysis.}
\begin{figure}
\centering
\begin{minipage}{0.6\textwidth}
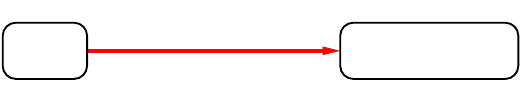\\[1ex]
\small
Job set: $\jobset=\{1,2\}$ with jobs $(0,M)$ and $(M-\epsilon,M+\epsilon)$\\
$1$-SORT schedule: $\tau_1,\tau_2,\pi_1,\pi_2$ with $\alg = 5M-\epsilon$\\
optimal schedule: $\tau_1,\pi_1,\tau_2,\pi_2$ with $\opt=  4M$
\end{minipage}
\caption{Instance with two jobs where the delay ratio on the arc $(1,2)$
is arbitrarily close to~$2$. A job with test time $t_j$ and processing time $p_j$ is written
as a pair $(t_j,p_j)$.}
\label{fig:intuition}
\end{figure}%
We consider an oriented complete graph $G=(V,A)$ where $V=\jobset$ and
each edge is directed towards the job with larger size.
We write $jk$ for the
arc (directed edge) from $j$ to~$k$. By~(\ref{eq:obj_alg}),
we can view the sum of completion times of a schedule as if it was produced
by a contribution $\sigma_j$ of each vertex $j\in V$ and a contribution
$D(j,k)$ of each arc $jk\in A$. The contributions of the vertices are the same
in the algorithm's schedule and in the optimal schedule. If
the \emph{delay ratio} $D(j,k)/D^*(j,k)$ is bounded by $\rho$
for every arc~$jk$, it follows that $\alg/\opt\le \rho$.
Unfortunately, the delay ratio of an individual arc can be arbitrarily
close to~$2$. Consider for example an instance with two jobs
with $t_1=0$, $p_1=M$ and $t_2=M-\epsilon$, $p_2=M+\epsilon$
for a large constant $M$ and
an infinitesimally small~$\epsilon>0$ (see Fig.~\ref{fig:intuition}).
Algorithm $1$-SORT schedules the operations in the order
$\tau_1,\tau_2,\pi_1,\pi_2$ giving $D(1,2)=t_1+p_1+t_2=2M-\epsilon$,
while the optimal schedule is $\tau_1, \pi_1,\tau_2,\pi_2$
with $D^*(1,2)=t_1+p_1=M$. Hence, the delay ratio
on the arc $(1,2)$ is arbitrarily close to~$2$. Nevertheless,
the ratio $\alg/\opt$ on this example does not exceed~$5/4$,
as the term $\sigma_1+\sigma_2=3M$ that makes the same contribution
to $\alg$ and $\opt$ is relatively large compared to the delays
on the arc $(1,2)$. We refer to arcs with large delay ratios
(to be defined precisely later on) as \emph{red} arcs. The
example suggests the idea of analyzing red arcs
together with other terms contributing
to the objective function in order to show a competitive ratio smaller than~$2$.

\begin{figure}
\centering
\begin{minipage}{0.8\textwidth}
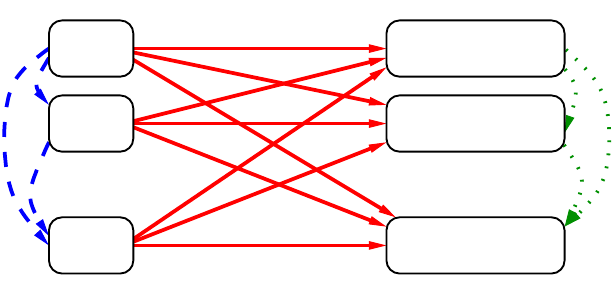\\[1ex]
\small
Job set: $\jobset$ contains jobs $1,2,\ldots,k$ of type $(0,M)$\\
\phantom{Job set:}
and jobs $k+1,k+2,\ldots,2k$ of type $(M-\epsilon,M+\epsilon)$\\
$1$-SORT schedule: $\tau_1,\ldots,\tau_{2k},\pi_1,\ldots,\pi_{2k}$ with $\alg \approx 4k^2M + kM$\\
optimal schedule: $\tau_1,\pi_1,\tau_2,\pi_2,\ldots,\tau_{2k},\pi_{2k}$, with $\opt\approx 2.5k^2M+1.5kM$
\end{minipage}
\caption{Instance with $2k$ jobs to illustrate red (drawn solid), blue (drawn
dashed) and green (drawn dotted) arcs.
 A job with test time $t_j$ and processing time $p_j$ is written
as a pair $(t_j,p_j)$.}
\label{fig:intuition2}
\end{figure}%

In the example of Fig.~\ref{fig:intuition} it was enough to consider
the red arc $(1,2)$ together with the contributions to the objective value made
by vertices $1$ and~$2$, but this kind of argument cannot suffice in general because
the number of arcs is quadratic in the number of vertices. Consider
the example of a job set with $n=2k$ jobs that contains $k$ copies of
each of the jobs from the previous example (see Fig.~\ref{fig:intuition2}). 
We call the $k$ jobs with $t_j=0$, $p_j=M$ \emph{left} jobs and the $k$ jobs
with $t_j=M-\epsilon$, $p_j=M+\epsilon$ \emph{right} jobs in the
following.
There are now $k^2$ arcs between left and right jobs, each with a delay ratio arbitrarily
close to~$2$. The contribution $k\cdot M+k\cdot 2M$ that the $2k$ vertices make
to the objective function is no longer sufficient to show a bound smaller
than $2$ for the competitive ratio, as it is negligible (for large $k$) compared to
the total delay on all the $k^2$ arcs between left and right jobs,
which is $k^2(2M-\epsilon)$ for $1$-SORT and $k^2M$ in the optimal schedule.
What we can exploit here instead is that the $k(k-1)/2$ arcs between left jobs
have the same delay $M$ in the schedule produced by $1$-SORT and in the
optimal schedule (delay ratio~$1$), and that the $k(k-1)/2$ arcs between right jobs have
delay $2M$ in the optimal schedule and delay $3M-\epsilon$ in the schedule
produced by $1$-SORT (delay ratio~$\approx 1.5$). We refer to the
arcs between left jobs as \emph{blue} arcs and to the arcs between right
jobs as \emph{green} arcs. The total delay on all the blue, red and green
arcs in this example is approximately $\frac{k^2}{2}M+k^2\cdot 2M+\frac{k^2}{2}\cdot 3M
= 4k^2M$ for $1$-SORT and approximately $\frac{k^2}{2}M+ k^2 M + \frac{k^2}{2}\cdot 2M =
2.5k^2M$ for the optimal schedule, where we have set $\epsilon=0$ and omitted terms linear in~$k$.
Thus, analyzing the red arcs together with the green and blue arcs
is sufficient to show that $\alg/\opt \le  4/2.5 = 1.6$ in this example.

\begin{figure}
\scalebox{0.8}{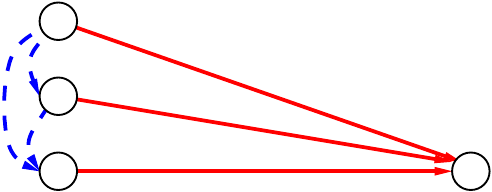}\hfill
\scalebox{0.8}{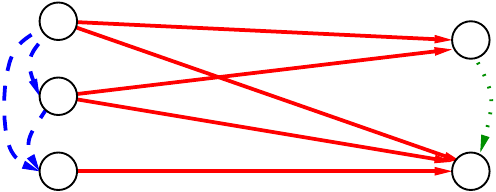}
\caption{Illustration of the idea underlying the analysis of a red vertex~$j$:
If the blue arcs have not yet been used in the
analysis of a previous red vertex, they can be used in
combination with the incoming red arcs of~$j$ (left). If blue
arcs have already been used in the analysis of a previous
red vertex $j'$, there must be a green arc between $j'$ and $j$
that is also available to be used in the analysis of the incoming
red arcs of~$j$.}
\label{fig:intuition3}
\end{figure}%
To turn these observations into a rigorous analysis, we will proceed as follows:
We give a formal definition of red arcs and refer to the vertices with incoming
red arcs as red vertices.
We then consider the red vertices in order of increasing~$t_j$. For
a red vertex~$j$, we would like to analyze the
delay of the incoming red arcs together with the blue arcs between their
tail vertices. If those blue arcs have not been used in the analysis of
previously considered red vertices, that suffices. If some of those blue arcs
have already been used in the analysis of previously considered red
vertices, however, then we can additionally use the green arcs that those previously
considered red vertices have to~$j$ in order to make up for the unavailability
of blue arcs. The crux of the analysis is a carefully specified invariant
that ensures that there are always sufficiently many green arcs available for
the analysis of a red vertex to make up for blue arcs that have been used in
the analysis of previously considered red vertices. See Fig.~\ref{fig:intuition3}
for an illustration of this idea. Overall, the outcome
is that the incoming red arcs of each red vertex can be analyzed together
with a sufficient number of blue and green arcs (which are not used
for the analysis of any other red vertex) to get a ratio smaller
than~$2$. One slight complication is that an arc may play the role
of a green arc for one red vertex and the role of a blue arc for another
red vertex, but we can handle this by treating that arc as a combination
of a distinct \emph{special} blue arc and a distinct \emph{special} green arc.

\begin{table}
    \caption{Overview of some of the notation used in the analysis of $1$-SORT}
    \label{tab:notation}
    \centering
    \begin{tabular}{|l|p{0.8\textwidth}|}
    \hline
    $\mu$ & parameter that determines which jobs are imbalanced\\
    $\nu$ & parameter that determines when an operation is `not much smaller' than another\\
    $G=(V,A)$ & auxiliary graph with the $n$ jobs as vertices \\
    $R$ & set of red arcs in $G$\\
    $V_R$ & vertices with incoming red arcs, ordered by $t_j$\\
    $V_I$ & vertices representing imbalanced jobs, ordered by $m_j$\\
	$I_k$ & set of incoming red arcs of vertex $k$\\
    $N^-(k)$ & vertices with outgoing red arc to vertex $k$\\
	$o(k)$ & cardinality of $N^-(k)$\\
	$N^-_{\ge j}(k)$ & subset of $N^-(k)$ consisting of $j$ and vertices after $j$ (in the order of $V_I$)\\
	$o_{\ge j}(k)$ & cardinality of $N^-_{\ge j}(k)$\\
    $P(k)$ & vertices $r$ that come before $k$ in $V_R$ and have a red arc from a vertex in $N^-(k)$ \\
    $G_k$ & subgraph of $G=(V,A)$ induced by $\{k\}\cup N^-(k)\cup P(k)$\\
	$\Gamma_k$ & green arcs between $k$ and vertices in $P(k)$\\
	$\Gamma_k^S$ & special green arcs between $k$ and vertices in $P(k)$\\
	$B_k$ & subset of still unused blue arcs between vertices in $N^-(k)$\\
    $C_k$ & set with $C_k\supseteq I_k$ that has total delay ratio smaller than $2$\\
    \hline
    \end{tabular}
\end{table}
\subparagraph{Formal analysis.} Having given an intuitive overview of the
ideas underlying our analysis of $1$-SORT, we now proceed to present the
technical details. We use two parameters $\mu>1$ and $\nu$ with $0<\nu<1$, satisfying
$\mu > 1/\nu$ and $1+\frac{1}{\mu}\le \nu+\nu^2$.
Intuitively,
the parameter $\mu$ determines which jobs we view as imbalanced (having
a large factor between test time and processing time),
and the parameter $\nu$ determines when we view a test time or processing
time to be `not much smaller' than another value (namely, when it
is at least $\nu$ times the other value). Table~\ref{tab:notation} gives
an overview of notation used in the proof.

\begin{theorem}
\label{th:betaSort}%
	The $\beta$-SORT algorithm with $\beta = 1$ has competitive ratio at most $\rho$
	with
	\begin{align*}
\rho =  &\max\Biggl\{
\frac{\nu+\nu^2+2+\frac{2}{\mu}}{\nu+\nu^2+1+\frac{1}{\mu}}, 
1+\frac{1}{2+\nu},
\frac{\frac{4}{\nu}+\frac{4}{\mu\nu}+\nu+\nu^2+1}{\frac{2}{\nu}+\frac{2}{\mu\nu}+\nu+\nu^2},\\
&\frac{\frac{4}{\nu}+\frac{4}{\nu\mu}+\nu+\frac{1}{\mu+1}}{\frac{2}{\nu}+\frac{2}{\mu\nu}+\nu},
\frac{4+\frac{5}{\mu}+\nu}{2+\frac{2}{\mu}+\nu},\\
&1+\frac{1}{\nu+\nu^2},
1+\frac{1}{\nu(\mu+1)},
\frac{2\mu+1}{\mu+1}, 1+\nu \Biggr\}
\end{align*}
\end{theorem}

We will prove Theorem~\ref{th:betaSort} in the remainder of this section.
Choosing $\mu$ and $\nu$ so as to minimize the ratio of Theorem~\ref{th:betaSort}
(computation done using Mathematica) yields the following corollary.

\begin{corollary}
\label{cor:betaSort}
The ratio $\rho$ of Theorem~\ref{th:betaSort} is minimized for
$\mu=\mu_0\approx 6.16277$ and $\nu=\nu_0\approx 0.860389$, yielding that $\beta$-SORT
with $\beta=1$ has competitive ratio at most $1.86039$.
Here $\mu_0$ is the only real root of the polynomial
$-2-8\mu-13\mu^2-11\mu^3-4\mu^4+\mu^5$ and $\nu_0=\frac{\mu_0}{\mu_0+1}$.
The ratio is $\rho=\frac{1+2\mu_0}{1+\mu_0}$.
\end{corollary}

Recall that $D(j,k)=d(j,k)+d(k,j)$ is the sum of the delays caused by jobs $j$ and $k$
on each other in the schedule produced by $1$-SORT, and
$D^*(j,k)=d^*(j,k)+d^*(k,j)$ is the sum of the delays caused by jobs $j$ and $k$
on each other in the optimal schedule.
As discussed in Section~\ref{sec:prelim}, the optimal schedule executes
the jobs in SPT order (with respect to their size), giving the objective
value stated in~Equation~(\ref{eq:obj_opt}).

Using infinitesimal perturbations of the test times and processing times
of the jobs that do not affect the schedule produced by $1$-SORT nor the
optimal schedule, we can assume without loss of generality that no two values
in the set of the test times, processing times, and sizes of all jobs are equal.
Therefore, when we compare any two such values, we can always assume that
strict inequality holds.

For the purpose of the analysis, we
create an auxiliary graph $G = (\nodes, \arcs)$, with $|\nodes| = n$ and 
$|\arcs| = \binom{n}{2} = \frac{n (n-1)}{2}$. Each vertex represents a job
(both the testing and processing operation), and there is a single arc between
any two vertices. The arc between vertices
$j$ and $k$ is directed towards
$k$ if $\total[j] < \total[k]$ and towards $j$ otherwise.
Recall that we write $jk$ for an arc directed
from $j$ to~$k$.
In addition, we associate with each arc $jk$
the values $D(j,k)$ and $D^*(j,k)$
that represent the pairwise delay between jobs $j$ and $k$
in the schedule produced by $1$-SORT and in the optimal schedule,
respectively, and the \emph{delay ratio} $\rho_{jk}=D(j,k)/D^*(j,k)$.

By~(\ref{eq:obj_alg}) and~(\ref{eq:obj_opt}), we have
$$
\alg = \sum_{j\in V} \total + \sum_{jk\in A} \pdelay[j,k]{}
$$
and
$$
\opt = \sum_{j\in V} \total + \sum_{jk\in A} \pdelay[j,k]{*}\,.
$$
Note that the first sum is the same in both expressions and therefore contributes to $\alg$ and $\opt$ in the same way, while the second sum, which 
represents the pairwise delays among all jobs, differs. 
As discussed earlier,
the difficulty when aiming to show competitive
ratio smaller than~$2$ is that
there may exist arcs $jk$
for which $D(j,k)$ can be arbitrarily close to $2\cdot D^*(j,k)$.
Hence, we cannot hope to prove a bound better than
$\rho_{jk}\le 2$ for all arcs $jk$,
and such a bound would only yield $\alg/\opt \le 2$.
As each job~$j$ contributes $\sigma_j$ to both $\alg$ and
$\opt$, we say that the \emph{delay ratio} of job~$j$,
denoted by $\rho_j$, is equal to~$1$.
In order to prove a competitive ratio better than~$2$, we
need to show that arcs $jk$ with delay ratio close to
$2$ can be analysed together with
arcs for which the delay ratio is much smaller than
$2$ and/or together
with vertices, for which we know that the delay ratio is~$1$.
This then yields that the ratios $\rho_{jk}$ and $\rho_j$
are bounded by a constant smaller than $2$ \emph{on average}.

It turns out that the ratio of an arc $jk$ can
be close to $2$ only if job $j$ is \emph{imbalanced}
and the test time of job $k$ is smaller but not much smaller
than $m_j=\max\{t_j,p_j\}$, and $p_k$
is not much smaller than $t_k$.

\begin{definition}
\label{def:imbalanced}%
A job $j$ is called \emph{imbalanced}
if $m_j=\max\{t_j,p_j\}\ge \mu \cdot \min\{t_j,p_j\}$,
for a fixed constant $\mu>1$.
\end{definition}

Arcs that may have a delay ratio close to $2$ are captured by
the following definition.
\begin{definition}
\label{def:red}%
An arc $jk$ is called a \emph{red} arc
if all of the following conditions hold:
\begin{itemize}
\item $j$ is imbalanced, and
\item $m_j\ge t_k \ge \nu \cdot m_j$, and
\item $p_k\ge \nu \cdot t_k$.
\end{itemize}
Here, $\nu$ is a constant with $0<\nu<1$ that satisfies
$\mu>\frac{1}{\nu}$ and $1+\frac{1}{\mu}\le \nu+\nu^2$.
\end{definition}

Note that only imbalanced jobs can have outgoing red arcs.
Let \red{} be the subset of arcs that are red.

\begin{lemma}
\label{lem:rededgejsmaller}%
If $jk$ is a red arc, then $\sigma_j\le \sigma_k$.
\end{lemma}

\begin{proof}
As $j$ is imbalanced, we have $\sigma_j=m_j+\min\{t_j,p_j\}
\le (1+\frac{1}{\mu})m_j$.
Furthermore, $t_k\ge \nu \cdot m_j$ and $p_k \ge \nu \cdot t_j$
imply $\sigma_k=t_k+p_k\ge (\nu+\nu^2)m_j$.
The claim therefore follows from
$1+\frac{1}{\mu}\le \nu+\nu^2$.
\end{proof}

\begin{figure}[htbp]
	\centering
	\includegraphics[width=\textwidth]{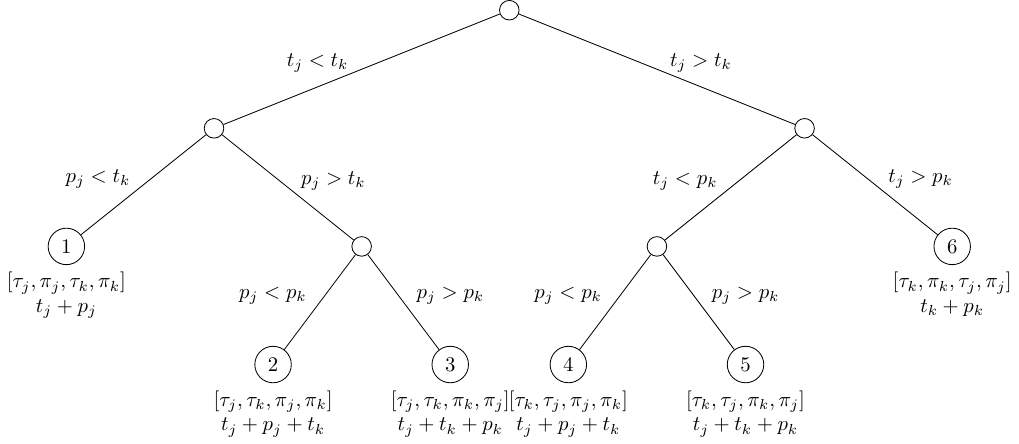}
\caption{Tree of possibilities for $D(j,k)$ and $D^*(j,k)$ and
$\rho_{jk}$ -- the top label at each leaf is the order of execution by
$1$-SORT,
the bottom label is $D(j,k)$. If we assume that $t_j+p_j\le t_k+p_k$,
then Leaf $5$ is impossible to reach and
$D^*(j,k)=t_j+p_j$ for all other leaves.}
\label{fig:treeposs}
\end{figure}

\begin{lemma}
\label{lem:notredratio}%
If an arc $jk$ is not red (i.e., $jk \notin \red$), then $\rho_{jk}$ is at most:
$$
\max\left\{
   \frac{2\mu+1}{\mu+1},
   1+\nu
\right\}
$$
\end{lemma}

\begin{proof}
For the following discussion, we refer to the tree of possibilities for $D(j,k)$ that
is shown in Fig.~\ref{fig:treeposs}.  
Leaf 5 cannot be reached if
$t_j+p_j<t_k+p_k$ (which holds as the arc $jk$ is directed from $j$ to~$k$)
as it is only reached if $t_j>t_k$ and $p_j>p_k$,
which would imply $t_j+p_j>t_k+p_k$. 
Leaf 1 gives a ratio of $1$, as the algorithm follows the optimal order.
Therefore, we only need to consider the leaves 2, 3, 4 and 6.
Note that the optimal pairwise delay is $\pdelay[j,k]{*} = \test + \proc$ for all cases.

First, assume that $j$ is not imbalanced. Recall that
$m_j=\max\{t_j,p_j\}$, and let $s_j=\min\{t_j,p_j\}$.
We have $s_j> m_j/\mu$ as $j$ is not imbalanced.
Thus, $t_j+p_j=m_j+s_j \ge m_j (1+1/\mu)$ and hence
\begin{equation}
m_j \le \frac{\mu}{\mu+1}(t_j+p_j).
\label{eq:xi}
\end{equation}
Now we consider the leaves 2, 3, 4 and~6.
\begin{itemize}
\item Leaf 2: As $t_j<t_k<p_j$ in this branch, we have
$D(j,k)=t_j+p_j+t_k\le t_j+2p_j$.
Furthermore, $p_j=m_j \le \frac{\mu}{\mu+1}(t_j+p_j)$ by~(\ref{eq:xi}).
Therefore, $D(j,k)\le (1+\frac{\mu}{\mu+1})(t_j+p_j)$ and
thus $\rho_{jk} \le \frac{2\mu+1}{\mu+1}$.
\item Leaf 3:
We again have $t_j<p_j$ and hence $p_j = m_j \le \frac{\mu}{\mu+1}(t_j+p_j)$ by~(\ref{eq:xi}).
Furthermore, we have $t_k<p_j$ and $p_k<p_j$, giving
$D(j,k)=t_j+t_k+p_k\le t_j+2p_j$, and we conclude $\rho_{jk} \le \frac{2\mu+1}{\mu+1}$
in the same way as for Leaf~2.
\item Leaf 4:
We have $D(j,k)=t_j+p_j+t_k\le 2t_j+p_j$ as $t_k<t_j$ in this branch.
If $t_j<p_j$, we have $t_j\le (t_j+p_j)/2$ and thus $2t_j+p_j=(t_j+p_j)+t_j \le 1.5(t_j+p_j)$
and $\rho_{jk}\le 1.5 \le \frac{2\mu+1}{\mu+1}$ (as $\mu > 1$).
If $t_j>p_j$, we have $t_j= m_j \le \frac{\mu}{\mu+1}(t_j+p_j)$ by~(\ref{eq:xi}).
Therefore, $D(j,k)\le (1+\frac{\mu}{\mu+1})(t_j+p_j)$ and
thus $\rho_{jk} \le \frac{2\mu+1}{\mu+1}$.
\item Leaf 6:
We have $D(j,k)=t_k+p_k\le 2t_j$ as $t_j>t_k$ and $t_j>p_k$ in this branch.
The case $t_j<p_j$ is impossible as we would have $t_j+p_j>t_k+p_k$.
Thus, we must have $t_j>p_j$ and hence $t_j= m_j \le \frac{\mu}{\mu+1}(t_j+p_j)$ by~(\ref{eq:xi}).
Therefore, $D(j,k)\le \frac{2\mu}{\mu+1}(t_j+p_j)$ and 
thus $\rho_{jk} \le \frac{2\mu}{\mu+1}\le \frac{2\mu+1}{\mu+1}$.
\end{itemize}

For the remainder of the proof, assume that $j$ is imbalanced. Recall that $m_j=\max\{t_j,p_j\}$.
Assume that $t_k$ does not satisfy the condition
$m_j\ge t_k \ge \nu \cdot m_j$.
If $t_k>m_j$, we have $\rho_{jk}=1$ as the situation corresponds to Leaf~1 of
the tree in Fig.~\ref{fig:treeposs}.
Thus, assume that $t_k<\nu \cdot m_j$.
Again, we only need to consider the leaves 2, 3, 4 and 6.
\begin{itemize}
\item Leaf 2 and 4: As $t_k<\nu \cdot m_j\le \nu(t_j+p_j)$,
we have $D(j,k)=t_j+p_j+t_k \le (1+\nu)(t_j+p_j)$ and hence $\rho_{jk}\le 1+\nu$.
\item Leaf 3: As $p_k<p_j$ in this branch,
we have $D(j,k)=t_j+t_k+p_k \le t_j+p_j+t_k$ and we get $\rho_{jk}\le 1+\nu$ in
the same way as for Leaf~2 and~4.
\item Leaf 6: As $p_k<t_j$ in this branch,
we have $D(j,k)=t_k+p_k\le t_j+t_k\le t_j+p_j+t_k$ and we get $\rho_{jk}\le 1+\nu$ in the
same way as for Leaf~2 and~4.
\end{itemize}

From now on, assume that $j$ is imbalanced and
$m_j\ge t_k \ge \nu \cdot m_j$ holds.
Assume that the condition $p_k \ge \nu \cdot t_k$ is violated,
so that we have $p_k<\nu \cdot t_k$.
Again, we only need to consider the leaves 2, 3, 4 and 6.
\begin{itemize}
\item Leaf 2: This branch is not possible: We have $p_k>p_j>t_k$ in this
branch, so $p_k<\nu \cdot t_k$ with $0 < \nu < 1$ cannot hold.
\item Leaf 3: As $t_k<p_j$ in this branch,
we have $D(j,k)=t_j+t_k+p_k<t_j+t_k+\nu \cdot t_k=t_j+(1+\nu)t_k<t_j+(1+\nu)p_j
\le (1+\nu)(t_j+p_j)$ and hence $\rho_{jk}\le 1+\nu$.
\item Leaf 4: This branch is not possible: We have $p_k>t_j>t_k$ in this
branch, so $p_k<\nu \cdot t_k$ with $0 < \nu < 1$ cannot hold.
\item Leaf 6: As $t_k<t_j$ in this branch, we
have $D(j,k)=t_k+p_k< (1+\nu)t_k < (1+\nu) t_j \le (1+\nu)(t_j+p_j)$
and hence $\rho_{jk}\le 1+\nu$.
\end{itemize}
Thus, we have shown that in all cases where one of the conditions of Definition~\ref{def:red}
is violated, we have $\rho_{jk}\le\max\{\frac{2\mu+1}{\mu+1},1+\nu\}$.
\end{proof}

By considering the tree of possibilities for $D(j,k)$
shown in Fig.~\ref{fig:treeposs}, it is also easy to see
that $D(j,k)\le 2 D^*(j,k)$ holds for all arcs (including red
arcs).

In the following, we will show that, for each job with incoming
red arcs, those arcs can be grouped together with a set of non-red
arcs and the size of the job in such a way that the total ratio of
the algorithm's delay over the optimal delay for the group is bounded by a constant
$\rho$ that is smaller than~$2$.

Let $V_I$ be the set of all imbalanced jobs,
ordered by non-decreasing~$m_j$.
Let $V_R$ be the set of jobs with at least one incoming red arc.
(If $V_R$ is empty, the competitive ratio of the algorithm is bounded
by the ratio of Lemma~\ref{lem:notredratio}.)
Consider the jobs in $V_R$ to be sorted in order of non-decreasing test times 
and write $i\prec j$ if the test time of $i$ comes before the test
time of $j$ in that order.
Consider a particular job $k\in V_R$ with test time~$t_k$.
Every incoming red arc $jk$ of $k$ must come from a job
$j$ that is imbalanced and satisfies $m_j\ge t_k \ge \nu \cdot m_j$
and, as $j$ is imbalanced, $\min\{t_j,p_j\}\le m_j/\mu$. 
Let $N^-(k)$ be the set of vertices that have an outgoing red
arc to $k$, i.e., $N^-(k)=\{j\mid jk \in \red\}$.
For a vertex $j\in N^-(k)$, let $N^-_{\ge j}(k)$ be
the subset of $N^-(k)$ that consists of $j$ and all
vertices of $N^-(k)$ that come after $j$ (in the order of~$V_I)$.
Furthermore, let $P(k)$ denote the set of jobs coming
before $k$ in $V_R$ (in the order $\prec$) that also have an incoming
red arc from at least one job in $N^-(k)$.

We process the jobs in $V_R$ in $\prec$-order.
To handle a job $k \in V_R$, we consider the subgraph $G_k$ of $G$
induced by $V_k=\{ k \}\cup N^-(k) \cup P(k)$. We call arcs
between two jobs in $N^-(k)$ \emph{blue} and
arcs between $k$ and any job in $P(k)$ \emph{green}.
The directions of blue and green arcs are irrelevant
and can be ignored. This means that when we refer to
a blue or green arc as~$jr$, that arc could be
directed from $j$ to $r$ or from $r$ to~$j$.
We denote by $C_k$ the set of elements
(vertices and arcs) of $G$ that are grouped with the
red incoming arcs of $k$ for the analysis. We will
always have that $k$ and its incoming red arcs are
in $C_k$, and we will add a suitable
number of blue and/or green arcs to~$C_k$.
Each blue and/or green arc will be added to at most
one such set~$C_k$, except in a special case where
an arc $e$ plays the role of a blue arc for one $k$
and the role of a green arc for another $k$;
in that case, we will split $e$ into a green arc
and a blue arc, and each part will be added to
at most one set~$C_k$.
We let $\rho_{C_k}$ denote the ratio of the sum of the delays
on all the arcs and vertices in $C_k$ in the solution
by the algorithm divided by the sum of the delays on the
same arcs and vertices in the optimal schedule.

First, we observe that, for any job $k$ that has incoming red arcs,
the set $N^-(k)$ is a contiguous subset of~$V_I$.

\begin{lemma}
\label{lem:interval}
For any job $k$ that has incoming red arcs,
the set $N^-(k)$ is a contiguous subset of~$V_I$.
\end{lemma}

\begin{proof}
Let $j_1$ and $j_2$ be the jobs in $V_I$
that minimize and maximize $m_j$ among all
jobs $j$ in $N^-(k)$, respectively.
As $j_1k$ and $j_2k$ are red arcs,
we have $t_k\le m_{j_1}$, $t_k\ge \nu \cdot m_{j_2}$,
and $p_k\ge \nu \cdot t_k$ (by Definition~\ref{def:red}).
Let $j$ be any job with $m_{j_1}\le m_j \le m_{j_2}$.
It follows that $m_j \ge m_{j_1}\ge t_k$ and
$\nu \cdot m_j\le \nu \cdot m_{j_2}\le t_k$.
Hence, $jk$ is a red arc and $j$ is also in $N^-(k)$.
\end{proof}

\begin{lemma}
\label{lem:intprefix}
Let job $k$ be a job with incoming red arcs,
and let $r\in P(k)$.
Then the intersection of $N^-(k)$ and $N^-(r)$
is a (not necessarily proper) prefix of $N^-(k)$.
Furthermore, $N^-(r)$ cannot contain any vertex
in $V_I$ that comes after $N^-(k)$.
\end{lemma}

\begin{proof}
Let $j_1$ be the job that minimizes $m_{j}$ among
all $j\in N^-(k)$.
Let $j$ be a job in $N^-(k)\cap N^-(r)$.
Note that $m_j\ge m_{j_1}$.
We need to show that $j_1\in N^-(r)$.

As $jr$ is a red arc,
we have $m_j\ge t_r \ge \nu \cdot m_j$ and $p_r \ge \nu \cdot t_r$.
With $m_j\ge m_{j_1}$, this implies $t_r\ge \nu \cdot m_{j_1}$.
As $r\in P(k)$, we have $t_r\le t_k$.
As $j_1k$ is a red arc,
we have $m_{j_1}\ge t_k \ge \nu \cdot m_{j_1}$.
With $t_r\le t_k$, this implies $m_{j_1}\ge t_r$.
Therefore, $j_1r$ is also
a red arc, so $j_1\in N^-(r)$. This shows that
$N^-(k)\cap N^-(r)$ is a prefix of $N^-(k)$.

Now, assume that $N^-(r)$ contains a vertex $j$ that
comes after the last vertex of $N^-(k)$ in $V_I$.
This means that $m_j\ge t_r \ge \nu \cdot m_j$ and
$p_r\ge \nu \cdot t_r$.
As $r$ was processed before $k$ in $V_R$, we
have $t_r\le t_k$.
Hence, $\nu \cdot m_j\le t_r$ implies $\nu \cdot m_j\le t_k$.
Let $j'$ be an arbitrary vertex in $N^-(k)$.
As $m_j\ge m_{j'}$ and $t_k\le m_{j'}$, we also have $t_k\le m_j$.
Therefore, $jk$ is also a red arc, a contradiction
to $j$ coming after $N^-(k)$ in~$V_I$.
\end{proof}

We say that a blue arc is \emph{used} or \emph{used up} in the analysis
of a vertex $k\in V_R$ if the arc is added to the set~$C_k$.
We maintain the following invariant when processing the
vertices in $V_R$.

\begin{invariant}
\label{inv:full}%
Consider a vertex $k\in V_R$, and any vertex $j$ in $N^-(k)$.
Let $P_{\ge j}(k)$ be the set of vertices in $V_R$ that have
been processed before $k$ and that have a red arc from~$j$.  
For each $r\in \{k\}\cup P_{\ge j}(k)$,
let $o_{\ge j}(r) =|N^-_{\ge j}(r)|$.
Then the total
number of blue arcs between vertices in $N^-_{\ge j}(k)$
that
have been used up in the analysis of vertices in $\{k\}\cup P_{\ge j}(k)$ at the time
when $k$ has just been processed
is at most $\sum_{r\in \{k\}\cup P_{\ge j}(k)}(o_{\ge j}(r)-1)$.
\end{invariant}

\begin{figure}
\centering
\resizebox{0.8\textwidth}{!}{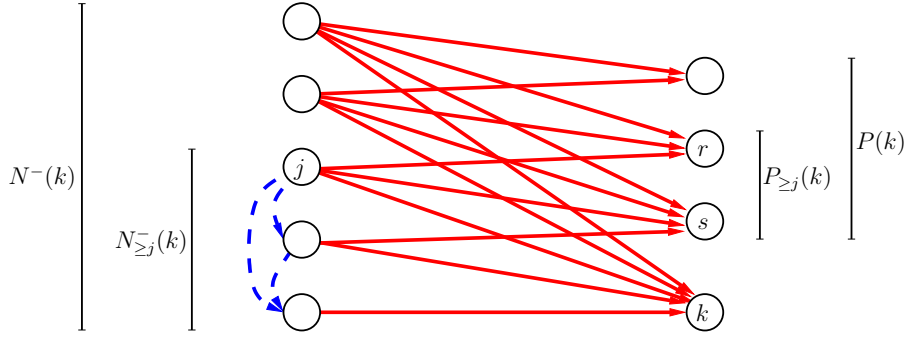}
\caption{Illustration of Invariant~\ref{inv:full}. Only blue arcs
between vertices in $N^-_{\ge j}(k)$ are shown.}
\label{fig:invariant}
\end{figure}%
Intuitively, if we imagine the vertices of $V_I$ arranged from top to bottom in
order of increasing $m_j$, the invariant says that for any vertex $j$ in $V_I$
the following condition holds: The number of blue arcs between vertices below
$j$ (including~$j$) that have been used up in the analysis of vertices in $V_R$
that have already been processed is bounded by the sum, over all those
vertices, of their number of incoming red arcs from vertices below $j$ minus one.
See Fig.~\ref{fig:invariant} for an illustration. In that figure,
vertex $r$ has $o_{\ge j}(r)=1$ incoming red arc from $N^-_{\ge j}(k)$, $s$~has
$o_{\ge j}(s)=2$, and $k$ has $o_{\ge j}(k)=3$. Therefore, the invariant says that,
after $k$ has been processed, the number of blue arcs between vertices
in $N^-_{\ge j}(k)$ that have been used up is at most $(1-1)+(2-1)+(3-1)=3$.

Note that Invariant~\ref{inv:full}
trivially holds before any vertices in $V_R$ are processed because
no blue arcs have been used at that point.

Before we prove that we can maintain Invariant~\ref{inv:full}
and construct sets $C_k$ that allow us to charge the incoming red arcs
of each vertex $k\in V_R$ to blue and green arcs and $k$ itself, we establish some
properties of blue and green arcs.

\begin{lemma}
\label{lem:blue}
For each blue arc $ij$ in $G_k$,
we have $D^*(i,j)\ge t_k$
and $\rho_{ij} \le \rho^B = 1+\frac{1}{\mu\nu}$.
\end{lemma}

\begin{proof}
Assume w.l.o.g.\ that $\sigma_i \le \sigma_j$.
Then $D^*(i,j)= \sigma_i \ge m_i \ge t_k$,
where the last inequality follows as $ik$ is a red arc.

Since $ik$ and $jk$ are red arcs,
we have $m_i\ge t_k\ge \nu \cdot m_i$ and
$m_j\ge t_k\ge \nu \cdot m_j$.
Therefore, $m_i\le \frac{1}{\nu} t_k\le \frac{1}{\nu} m_j$
and $m_j\le \frac{1}{\nu} t_k\le \frac{1}{\nu} m_i$.
This also shows $\max\{m_i,m_j\}\le \frac{1}{\nu}\min\{m_i,m_j\}$.

Let $s_i=\min\{t_i,p_i\}$ and
$s_j=\min\{t_j,p_j\}$. As $s_i\le \frac{m_i}{\mu}$ (job $i$ is imbalanced)
and $m_i\le \frac{1}{\nu}m_j$, we have
$s_i\le \frac{1}{\mu\nu} m_j\le m_j$, as
$\mu\ge 1/\nu$. Similarly,
$s_j\le m_i$.
We now consider cases depending on whether
$s_i=t_i$ or $s_i=p_i$ and $s_j=t_j$ or $s_j=p_j$
and show that $D(i,j)\le D^*(i,j)+s_j$ holds
in every case:
\begin{itemize}
\item $s_i=t_i$ and $s_j=t_j$:
The algorithm will first execute the tests with times $s_i$ and $s_j$ (in some order),
and then the processing parts with times $m_i$ and $m_j$ (in some order).
We have
$D(i,j)= s_i+s_j+\min\{m_i,m_j\}$.
As $D^*(i,j)=s_i+m_i$,
we have $D(i,j)\le D^*(i,j)+s_j$.

\item $s_i=t_i$, $s_j=p_j$:
The algorithm will execute $\tau_i$ first and then
either $\pi_i$ or $\tau_j$. In the former case,
we have $D(i,j)=s_i+m_i=D^*(i,j)$.
In the latter case we must have $t_j\le p_i$ and
the schedule is $[\tau_i,\tau_j,\pi_j,\pi_i]$.
We have $D(i,j)=t_i+t_j+p_j\le t_i+p_i+p_j
\le D^*(i,j)+s_j$.

\item $s_i=p_i$, $s_j=t_j$:
The algorithm will execute $\tau_j$ first and
then either $\pi_j$ or $\tau_i$. In the former
case, we must have $p_j\le t_i$. The schedule
is $[\tau_j,\pi_j,\tau_i,\pi_i]$
and we have $D(i,j)=t_j+p_j \le t_j + t_i \le
D^*(i,j)+s_j$.
In the latter case, the schedule
is $[\tau_j,\tau_i,\pi_i,\pi_j]$
and we have $D(i,j)=t_i+p_i+t_j=D^*(i,j)+t_j = D^*(i,j)+s_j$.

\item $s_i=p_i$, $s_j=p_j$:
If $t_i<t_j$, the schedule is
$[\tau_i,\pi_i,\tau_j,\pi_j]$ and
we have $D(i,j)=D^*(i,j)$.
If $t_i>t_j$, the schedule is
$[\tau_j,\pi_j,\tau_i,\pi_i]$ and
we have $D(i,j)=t_j+p_j\le t_i+p_j
\le t_i+p_i+p_j = D^*(i,j)+s_j$.
\end{itemize}
In all cases, we have $D(i,j)\le D^*(i,j)+s_j$.
As $s_j\le \frac{1}{\mu}m_j \le \frac{1}{\mu\nu}m_i$,
we get $D(i,j)\le (1+\frac{1}{\mu\nu})D^*(i,j)$.
\end{proof}

\begin{lemma}
\label{lem:green}
For each green arc $jk$ (with $j\in P(k)$) in $G_k$,
we have $D^*(j,k)\ge (\nu+\nu^2) t_k$
and $\rho_{jk} \le \rho^G = 1+\frac{1}{\nu+\nu^2}$.
\end{lemma}

\begin{proof}
Let $jk$ be a green arc such that $j\in P(k)$, and let $j'$
be a vertex in $N^-(k)$ that has a red arc
to both $j$ and~$k$.
By the definition of red arcs,
the values $t_r$ and $t_k$ are at least
$\nu m_{j'}$ and at most $m_{j'}$. Therefore
$t_k\ge t_j\ge \nu \cdot t_k$. 
(The first inequality holds because job $j$ is processed before job $k$.)
Furthermore, we have $p_j\ge \nu \cdot t_j$
and $p_k\ge \nu \cdot t_k$ as $j,k\in V_R$.
Therefore, $D^*(j,k)=\min\{t_k+p_k,t_j+p_j\}
\ge \min\{t_k,t_j\}+\min\{p_k,p_j\}
\ge \nu t_k + \min\{\nu t_k,\nu^2 t_k\}
= \nu t_k + \nu^2 t_k = (\nu+\nu^2)t_k$.

As $j\in P(k)$, we have
$t_j\le t_k$, and hence only the left half
of the decision tree in Fig.~\ref{fig:treeposs}
is relevant.
Let $i$ be a vertex in $N^-(j)\cap N^-(k)$.
As $ij$ and $ik$ are red arcs, we have
$m_i \ge t_j \ge \nu \cdot m_i$ and
$m_i \ge t_k \ge \nu \cdot m_i$
and $p_j\ge \nu \cdot t_j$ and $p_k\ge \nu \cdot t_k$.

The schedule produced by the
algorithm can be as in leaves 1, 2 and 3 of the tree,
and we consider each case as follows:
\begin{itemize}
\item Leaf 1: We have $D(j,k)=t_j+p_j$.
If $\sigma_j\le \sigma_k$, then $D^*(j,k)=t_j+p_j=D(j,k)$ and $\rho_{jk}=1$.
Otherwise, $D^*(j,k) = t_k + p_k \ge (1+\nu)t_k$, while
$D(j,k)=t_j+p_j\le 2 t_k$ (because $t_j<t_k$ and $p_j<t_k$ hold
in the branch leading to Leaf~$1$).
Therefore, $\rho_{jk}\le \frac{2}{1+\nu}$.

\item Leaf 2: We have $D(j,k)=t_j+p_j+t_k$.
As $t_j<t_k$ and $p_j<p_k$ in this branch,
we have $\sigma_j\le \sigma_k$ and
therefore $D^*(j,k)=t_j+p_j$.
As $t_k\le m_i$ and $t_j+p_j\ge t_j+\nu \cdot t_j = (1+\nu)t_j \ge (\nu+\nu^2) m_i$,
we have $t_k\le \frac{1}{\nu+\nu^2}(t_j+p_j)$ and
hence $D(j,k)\le (1+\frac{1}{\nu+\nu^2}) D^*(j,k)$.
This shows $\rho_{jk}\le 1+\frac{1}{\nu+\nu^2}$.

\item Leaf 3: We have $D(j,k)=t_j+t_k+p_k$.
If $\sigma_j\le \sigma_k$, then $D^*(j,k)=t_j+p_j$.
As we have $p_k< p_j$ in this branch,
$D(j,k)\le t_j+t_k+p_j$.
Furthermore, $t_k\le m_i$
and $t_j+p_j\ge (\nu+\nu^2) m_i$ (same as for Leaf~2), so
$t_k\le \frac{1}{\nu+\nu^2}(t_j+p_j)$.
Hence, $D(j,k)\le (1+\frac{1}{\nu+\nu^2}) D^*(j,k)$
and $\rho_{jk}\le 1+\frac{1}{\nu+\nu^2}$.

If $\sigma_j > \sigma_k$, then $D^*(j,k)=t_k+p_k$.
As $t_j\le m_i$ and $t_k+p_k\ge (1+\nu)t_k \ge (\nu+\nu^2) m_i$,
we have $t_j \le \frac{1}{\nu+\nu^2}(t_k+p_k)$.
This gives $D(j,k)=t_j+t_k+p_k\le (1+\frac{1}{\nu+\nu^2})(t_k+p_k)=
(1+\frac{1}{\nu+\nu^2}) D^*(j,k)$. Thus, $\rho_{jk}\le 1+\frac{1}{\nu+\nu^2}$.
\end{itemize}
The bounds on the ratios in all cases are bounded by
$\max\{\frac{2}{1+\nu},1+\frac{1}{\nu+\nu^2}\}=1+\frac{1}{\nu+\nu^2}$.
\end{proof}

Unfortunately, it is possible that an arc $ij$ is used as a blue arc
in the analysis of one vertex $r$ in $V_R$ and as a green arc in the analysis
of another vertex $k$ in $V_R$. We handle this case by splitting such an arc $ij$
into two arcs for the purpose of the analysis, a \emph{special blue arc}
used in the analysis of $r$ and a \emph{special green arc} used in the
analysis of $k$. In this way, we can ensure that every arc is used in the
analysis of at most one vertex.

Consider an arc $ij$ that is used both as a blue and a green arc.
As $ij$ is used as a blue arc, both $i$ and $j$ must be imbalanced.
As $ij$ is used as a green arc, we must have $p_i\ge \nu \cdot t_i$ and $p_j\ge \nu \cdot t_j$.
As $i$ is imbalanced, we must have either $p_i\ge \mu t_i$
or $p_i\le t_i / \mu$.
As the condition $\mu > 1/\nu$ from Definition~\ref{def:red}
is equivalent to $\nu> 1/\mu$, we
have $p_i\ge \nu \cdot t_i> t_i/\mu$.
Thus, the case $p_i \le t_i/\mu$ is
not possible, and so we must have
$p_i\ge \mu t_i$.
Similarly, $p_j\ge \mu t_j$.
Therefore, $m_i=p_i$ and $m_j=p_j$.
As $ij$ is a blue arc, there is a job $k\in V_R$ such that $ik$ and $jk$ are red
arcs. Therefore, by Definition~\ref{def:red},
$m_i\ge t_k\ge \nu m_i$ and $m_j\ge t_k\ge \nu m_j$ .
This shows that both $m_i$ and $m_j$ lie
in the interval $[t_k,\frac{1}{\nu}t_k]$.
Therefore, $m_i=p_i$ and $m_j=p_j$ can differ
at most by a factor of $\nu$, which means
$\min\{p_i,p_j\}\ge \nu \max\{p_i,p_j\}$.
As $ij$ is a green arc, there is a job $s\in V_I$
such that $si$ and $sj$ are red arcs.
Therefore $m_s \ge t_i \ge \nu \cdot m_s$ and
$m_s \ge t_j \ge \nu \cdot m_s$, so
$t_i$ and $t_j$ both lie in the interval
$[\nu\cdot m_s,m_s]$ and hence
can also differ at most by a factor of $\nu$,
giving $\min\{t_i,t_j\}\ge \nu \max\{t_i,t_j\}$.

Assume without loss of generality that $t_i<t_j$.
We have $D^*(i,j)\ge \min\{t_i,t_j\}+\min\{p_i,p_j\}$
and $D(i,j)=t_i+t_j+\min\{p_i,p_j\}$.
Note that $t_i+t_j\le (1+\frac{1}{\nu})\min\{t_i,t_j\}$
and $\min\{t_i,t_j\}\le \frac{1}{\mu}\min\{p_i,p_j\}$.
We have $$
\rho_{ij}\le \frac{(1+\frac{1}{\nu})\min\{t_i,t_j\} +\min\{p_i,p_j\}}{\min\{t_i,t_j\}+\min\{p_i,p_j\}} \,.
$$
As $1+\frac{1}{\nu}>1$, the ratio is maximized when $\min\{t_i,t_j\}$ is as large as possible,
so we set $\min\{t_i,t_j\}=\frac{1}{\mu}\min\{p_i,p_j\}$ and obtain
$$
\rho_{ij}\le \frac{(1+\frac{1}{\nu})\frac{1}{\mu}\min\{p_i,p_j\} +\min\{p_i,p_j\}}{\frac{1}{\mu}\min\{p_i,p_j\}+\min\{p_i,p_j\}}
= \frac{1+\frac{1}{\mu}+\frac{1}{\mu\nu}}{1+\frac{1}{\mu}}
= 1+\frac{1}{\nu(\mu+1)}
$$
We split the arc $ij$ into two arcs (denoted by $i_gj_g$ and $i_bj_b$)
as follows:
\begin{itemize}
\item A \emph{special green arc} $i_gj_g$ with $D^*(i_g,j_g) \ge \min\{t_i,t_j\}$.
\item A \emph{special blue arc} $i_bj_b$ with $D^*(i_b,j_b) \ge \min\{p_i,p_j\}$.
\item Choose $D^*(i_g,j_g)$ and $D^*(i_b,j_b)$ in such a way that they add up to $D^*(i,j)$.
\item Split $D(i,j)$ between $D(i_g,j_g)$ and $D(i_b,j_b)$ in the same proportion as
$D^*(i,j)$ has been split into $D^*(i_g,j_g)$ and $D^*(i_b,j_b)$, so
that $\rho_{igj_g}=\rho_{i_bj_b}=\rho_{ij}\le 1+\frac{1}{\nu(\mu+1)}$
\end{itemize}

We then get the following lemmas for special blue and green arcs.

\begin{lemma}
\label{lem:specialblue}
For each special blue arc $i_bj_b$ in $G_k$,
we have $D^*(i_b,j_b)\ge t_k$
and $\rho_{i_bj_b} \le \rho^S = 1+\frac{1}{\nu(\mu+1)}$.
\end{lemma}

\begin{proof}
The bound on $\rho_{i_bj_b}$ was shown above.
As $p_i\ge t_k$ and $p_j\ge t_k$ and
$D^*(i_b,j_b)\ge \min\{p_i,p_j\}$ by definition,
we have $D^*(i_b,j_b)\ge t_k$.
\end{proof}

Lemma~\ref{lem:specialblue} shows that a special blue arc satisfies the properties
of a regular (non-special) blue arc stated in Lemma~\ref{lem:blue}, and hence
we do not distinguish between blue arcs and special blue arcs in the remainder
of the proof.

\begin{lemma}
\label{lem:specialgreen}
For each special green arc $r_gk_g$ (with $r\in P(k)$) in $G_k$,
we have $D^*(r_g,k_g)\ge \nu t_k$
and $\rho_{r_gk_g} \le \rho^S = 1+\frac{1}{\nu(\mu+1)}$.
\end{lemma}

\begin{proof}
The bound on $\rho_{r_gk_g}$ was shown above.
As $rk$ is a green arc, there is a vertex $s\in V_I$
such that $sr$ and $sk$ are red arcs, which implies
that $t_r=\min\{t_r,t_k\}\ge \nu \max\{t_r,t_k\}=t_k$ (as job $r$ is processed before job $k$).
As $D^*(r_g,k_g)\ge \min\{t_r,t_k\}\ge \nu t_k$ by definition,
we have $D^*(r_g,k_g)\ge \nu t_k$.
\end{proof}

The following lemma deals with vertices in $V_R$ that have a single
incoming red arc. The lemma shows that we do not need to use any
blue arcs for such vertices, so they do not play any role in the
process of maintaining Invariant~\ref{inv:full}.

\begin{lemma}
\label{lem:onered}
If $|N^-(k)|=1$ and we take $C_k=\{k,jk\}$, where $jk$ is the
single incoming red arc of~$k$, then the ratio $\rho_{C_k}$ of the
algorithms's delay over the optimal delay in $C_k$ is
bounded by $\frac{\nu+\nu^2+2+\frac{2}{\mu}}{\nu+\nu^2+1+\frac{1}{\mu}}$.
\end{lemma}

\begin{proof}
The delay of $k$ on itself is $\total[k]$ both in the optimal
schedule and in the algorithm's schedule.
Note that $\total[k]=t_k+p_k\ge (1+\nu)t_k \ge \nu(1+\nu)m_j$.
The delay $D^*(j,k)$ is $\total[j]\le m_j+\frac{m_j}{\mu}=(1+\frac{1}{\mu})m_j$.

The ratio of the delays in $C_k$ is then bounded by:
$$
\frac{\total[k] + D(j,k)}{\total[k]+D^*(j,k)}
\le 
\frac{\total[k] + 2D^*(j,k)}{\total[k]+D^*(j,k)}
$$
This ratio is maximized when $\total[k]$ is as small
as possible and $D^*(j,k)$ is as large
as possible, so we can set $\total[k]=(\nu+\nu^2) m_j$ and
$D^*(j,k)=(1+\frac{1}{\mu})m_j$, giving:
$$
\frac{\total[k] + 2D^*(j,k)}{\total[k]+D^*(j,k)}\le
\frac{(\nu+\nu^2) m_j + 2 (1+\frac{1}{\mu})m_j}{(\nu+\nu^2) m_j+(1+\frac{1}{\mu})m_j}=
\frac{(\nu+\nu^2) + 2 (1+\frac{1}{\mu})}{(\nu+\nu^2) +(1+\frac{1}{\mu})}
$$
This shows the bound of the lemma.
\end{proof}

Now we show for each vertex $k$ in $V_R$ that, assuming Invariant~\ref{inv:full} holds before
vertex $k$ is processed, we can
construct a set $C_k$ that allow us to charge the red incoming arcs
while maintaining Invariant~\ref{inv:full}.

\begin{lemma}
\label{lem:inv}%
Let $k$ be a vertex in $V_R$ and assume that Invariant~\ref{inv:full} holds just
before $k$ is processed. We can define a set $C_k$ consisting of blue arcs
connecting vertices in $N^-(k)$, green arcs connecting $k$ with vertices in $P(k)$,
all incoming red arcs of~$k$, and $k$ itself in such a way that
$\rho_{C_k} \le \rho^C$ with
\begin{align*}
\rho^C = & \max\{
\frac{\nu+\nu^2+2+\frac{2}{\mu}}{\nu+\nu^2+1+\frac{1}{\mu}},
1+\frac{1}{2+\nu},
\frac{\frac{4}{\nu}+\frac{4}{\mu\nu}+\nu+\nu^2+1}{\frac{2}{\nu}+\frac{2}{\mu\nu}+\nu+\nu^2},\\
& \frac{\frac{4}{\nu}+\frac{4}{\nu\mu}+\nu+\frac{1}{\mu+1}}{\frac{2}{\nu}+\frac{2}{\mu\nu}+\nu},
\frac{4+\frac{5}{\mu}+\nu}{2+\frac{2}{\mu}+\nu},
\rho^G,\rho^S \}
\,.
\end{align*}
Furthermore, Invariant~\ref{inv:full} still holds after $k$
is processed.
\end{lemma}

\begin{proof}
We can assume $|N^-(k)|\ge 2$ as the case $|N^-(k)|=1$ can be
handled by Lemma~\ref{lem:onered} and does not use any blue arcs.

Let $I_k=\{k\}\cup \{jk\mid j\in N^-(k)\}$ be the set consisting of $k$
and all its incoming red arcs. The set $C_k$ will always be a superset
of~$I_k$, i.e., $C_k \supseteq I_k$.
Let $o(k)=|N^-(k)|$.

\noindent\textbf{Case 1:}
$|P(k)|\ge \frac{o(k)-1}{2}$.
In this case, we do not use any blue arcs and instead use only green arcs.
Let $\Gamma_k$ and $\Gamma_k^S$ be the green arcs and special green arcs, respectively,
in the set $\{kr\mid r\in P(k)\}$ of arcs between $k$
and vertices in $P(k)$. Note that $\Gamma_k\cup \Gamma_k^S=\{kr\mid r\in P(k)\}$.
We define $C_k=I_k\cup \Gamma_k\cup \Gamma_k^S$.
The ratio $\rho_{C_k}$ can be bounded as follows:
\begin{align*}
\rho_{C_k} &=
\frac{
\sigma_k + \sum_{j\in N^-(k)} D(j,k) + \sum_{rk\in \Gamma_k} D(r,k) + \sum_{rk\in \Gamma_k^S} D(r,k)
}{
\sigma_k + \sum_{j\in N^-(k)} D^*(j,k) + \sum_{rk\in \Gamma_k} D^*(r,k) + \sum_{rk\in \Gamma_k^S} D^*(r,k)
}
\\
&\le 
\frac{
\sigma_k + \sum_{j\in N^-(k)} 2 D^*(j,k) + \sum_{rk\in \Gamma_k} \rho^G D^*(r,k) + \sum_{rk\in \Gamma_k^S} \rho^S D^*(r,k)
}{
\sigma_k + \sum_{j\in N^-(k)} D^*(j,k) + \sum_{rk\in \Gamma_k} D^*(r,k) + \sum_{rk\in \Gamma_k^S} D^*(r,k)
}
\end{align*}
We use that $\frac{A+B}{C+D}\le \max\{A/C,B/D\}$ for $A,B,C,D>0$ (and the generalization
to sums with more than two terms in the enumerator and in the denominator)
and consider parts of the fraction as follows:
\begin{itemize}
\item For $k$ together with an arbitrary red arc $jk$, the ratio is
$\frac{t_k+p_k + 2D^*(j,k)}{t_k+p_k+D^*(j,k)}$.
As $D^*(j,k)\ge \min \{m_j,t_k\}\ge t_k$,
the ratio is bounded by
$\frac{t_k+p_k + 2t_k}{t_k+p_k+t_k}
=\frac{p_k+3t_k}{p_k+2t_k}$.
As $p_k\ge \nu t_k$, this ratio is bounded by
$\frac{(3+\nu)t_k}{(2+\nu)t_k}=1+\frac{1}{2+\nu}$.

\item We form groups of two red arcs with one green arc arbitrarily.
As $o(k)-1$ red arcs are left and the number of green arcs (special or not) is at least
$\frac{o(k)-1}{2}$, there are enough green arcs for this.
(If the last group consists of a single red arc together with a green
arc, the bound on the ratio is only better.)
The ratio for a group consisting of red arcs $ik$, $jk$ and
a non-special green arc $rk$ is bounded by
$\frac{2D^*(i,k)+2D^*(j,k)+\rho^G D^*(r,k)}{
D^*(i,k)+D^*(j,k)+D^*(r,k)}$.
As $\rho^G < 2$, the ratio is maximized when
$D^*(i,k)$ and $D^*(j,k)$ are as large as possible and $D^*(r,k)$
is as small as possible.
By Lemma~\ref{lem:green}, $D^*(r,k)\ge (\nu+\nu^2) t_k$.
Furthermore, by Lemma~\ref{lem:rededgejsmaller}
we know $\sigma_i\le \sigma_k$. Hence,
$D^*(i,k)=t_i+p_i\le m_i(1+\frac{1}{\mu}) \le t_k(\frac{1}{\nu} + \frac{1}{\mu\nu})$,
and similarly $D^*(j,k)\le t_k(\frac{1}{\nu} + \frac{1}{\mu\nu})$.
Hence, we get:
\begin{align*}
\frac{2D^*(i,k)+2D^*(j,k)+\rho^G D^*(r,k)}{D^*(i,k)+D^*(j,k)+D^*(r,k)}
&\le
\frac{4 t_k(\frac{1}{\nu} + \frac{1}{\mu\nu}) + \rho^G (\nu+\nu^2)t_k}{
2 t_k(\frac{1}{\nu} +\frac{1}{\mu\nu}) + (\nu+\nu^2) t_k}\\
&=\frac{4(\frac{1}{\nu} + \frac{1}{\mu\nu})+(\nu+\nu^2)\rho^G}{
2 (\frac{1}{\nu} +\frac{1}{\mu\nu})+(\nu+\nu^2) }\\
&=\frac{\frac{4}{\nu}+\frac{4}{\mu\nu}+\nu+\nu^2+1}{\frac{2}{\nu}+\frac{2}{\mu\nu}+\nu+\nu^2} 
\end{align*}
The ratio for a group consisting of red arcs $ik$, $jk$ and
a special green arc $rk$ is bounded by
$\frac{2D^*(i,k)+2D^*(j,k)+\rho^S D^*(r,k)}{
D^*(i,k)+D^*(j,k)+D^*(r,k)}$.
As $\rho^S < 2$, the ratio is maximized when
$D^*(i,k)$ and $D^*(j,k)$ are as large as possible and $D^*(r,k)$
is as small as possible.
By Lemma~\ref{lem:specialgreen}, $D^*(r,k)\ge \nu t_k$.
Furthermore, by Lemma~\ref{lem:rededgejsmaller}
we know $\sigma_i\le \sigma_k$. Hence,
$D^*(i,k)=t_i+p_i\le m_i(1+\frac{1}{\mu}) \le t_k(\frac{1}{\nu} + \frac{1}{\mu\nu})$,
and similarly $D^*(j,k)\le t_k(\frac{1}{\nu} + \frac{1}{\mu\nu})$.
Hence, we get:
\begin{align*}
\frac{2D^*(i,k)+2D^*(j,k)+\rho^S D^*(r,k)}{D^*(i,k)+D^*(j,k)+D^*(r,k)}
&\le
\frac{4 t_k(\frac{1}{\nu} + \frac{1}{\mu\nu}) + \rho^S \nu t_k}{
2 t_k(\frac{1}{\nu} +\frac{1}{\mu\nu}) + \nu t_k}\\
&=\frac{4(\frac{1}{\nu} + \frac{1}{\mu\nu})+\nu\rho^S}{
2 (\frac{1}{\nu} +\frac{1}{\mu\nu})+\nu }\\
&=\frac{\frac{4}{\nu}+\frac{4}{\nu\mu}+\nu+\frac{1}{\mu+1}}{\frac{2}{\nu}+\frac{2}{\mu\nu}+\nu} 
\end{align*}

\item If $C_k$ contains additional green arcs $rk$, the part
$\frac{\rho^G D^*(r,k)}{D^*(r,k)}$ of the fraction is bounded
by $\rho^G$.
\item If $C_k$ contains additional special green arcs $rk$, the part
$\frac{\rho^S D^*(r,k)}{D^*(r,k)}$ of the fraction is bounded
by $\rho^S$.
\end{itemize}
Hence, $\rho_{C_k}\le \max\{1+\frac{1}{2+\nu},
\frac{\frac{4}{\nu}+\frac{4}{\mu\nu}+\nu+\nu^2+1}{\frac{2}{\nu}+\frac{2}{\mu\nu}+\nu+\nu^2},
\frac{\frac{4}{\nu}+\frac{4}{\nu\mu}+\nu+\frac{1}{\mu+1}}{\frac{2}{\nu}+\frac{2}{\mu\nu}+\nu}
\rho^G,\rho^S\}$.

As we have not used any blue arcs, Invariant~\ref{inv:full} is maintained
trivially in Case~1.

\noindent\textbf{Case 2:} $|P(k)|< \frac{o(k)-1}{2}$. 
If $P(k)\neq \emptyset$,
let $k'$ be the last vertex in $V_R$ that was processed before $k$,
and let $j$ be the first element (i.e., with smallest $m_j$) in $N^-(k)$.
Applying Invariant~\ref{inv:full} to $k'$ (as the $k$ in the statement
of the invariant) and $j$, we get that the number of blue arcs
between vertices in $N^-(k)$ that have been used up in the analysis
of vertices in $P(k)$ is at most
$\sum_{r\in P(k)} (o_{\ge j}(r)-1)
\le \sum_{r \in P(k)} (o(k)-1)$,
where $o_{\ge j}(r)\le o(k)$ follows from Lemma~\ref{lem:intprefix}.
Thus, at most $|P(k)|(o(k)-1)<\frac{(o(k)-1)^2}{2}$ blue arcs
have been used up.
If $P(k)=\emptyset$, no blue arcs between vertices in $N^-(k)$ have been used up.
The total number of blue arcs between vertices
in $N^-(k)$ is ${o(k) \choose 2} = \frac{o(k)(o(k)-1)}{2}$, so the
number of blue arcs that have not yet been used up is at least
$$
\frac{o(k)(o(k)-1)}{2}
-
\frac{(o(k)-1)^2}{2}
=
\frac{o(k)-1}{2}
\,.
$$
Hence, we have at least $\left\lceil \frac{o(k)-1}{2}\right\rceil$
unused blue arcs.
Again, let $\Gamma_k$ and $\Gamma_k^S$ with
$\Gamma_k\cup \Gamma_k^S=\{kr\mid r\in P(k)\}$ be the sets of green arcs and special
green arcs between $k$ and vertices in $P(k)$.
Let $z=o(k)-1-2|\Gamma_k\cup\Gamma_k^S|>0$ be the number of red arcs that cannot be
covered by one group consisting of a red arc and $k$ and $|\Gamma_k\cup \Gamma_k^S|$ groups
consisting of two red arcs and one (special or non-special) green arc each.
Let $B_k$ be a set of $\lceil z/2\rceil \le \left\lceil \frac{o(k)-1}{2}\right\rceil$
(special or non-special) blue arcs between vertices in $N^-(k)$ that have not yet been used up,
giving preference to arcs whose earlier endpoint (with respect to the
order of $V_I$) comes earlier in~$V_I$.
Define $C_k=I_k \cup \Gamma_k \cup \Gamma_k^S \cup B_k$.
Note that Invariant~\ref{inv:full} holds for $k$ and $j$
after $k$ is processed because $k$ will use up at most
$\lceil z/2\rceil \le z \le o(k)-1$ blue arcs. We will show
later that it also holds for all vertices $j'$ that come after $j$ in~$V_I$.
The ratio $\rho_{C_k}$ can be bounded as follows:
\begin{align*}
\rho_{C_k} &=
\frac{
\sigma_k + \sum_{j\in N^-(k)} D(j,k) + \sum_{rk\in \Gamma_k} D(r,k)+ \sum_{ij\in \Gamma_k^S} D(i,j) + \sum_{ij\in B_k} D(i,j)
}{
\sigma_k + \sum_{j\in N^-(k)} D^*(j,k) + \sum_{rk\in \Gamma_k} D^*(r,k)+ \sum_{ij\in \Gamma_k^S} D^*(i,j) + \sum_{ij\in B_k} D^*(i,j)
}
\\
&\le 
\frac{
\sigma_k + \sum_{j\in N^-(k)} 2 D^*(j,k) + \sum_{rk\in \Gamma_k} \rho^G D^*(r,k) + \sum_{ij\in \Gamma_k^S} \rho^S D^*(i,j) + \sum_{ij\in B_k} \rho^B D^*(i,j)
}{
\sigma_k + \sum_{j\in N^-(k)} D^*(j,k) + \sum_{rk\in \Gamma_k} D^*(r,k)+ \sum_{ij\in \Gamma_k^S} D^*(i,j) + \sum_{ij\in B_k} D^*(i,j)
}
\end{align*}
We again use that $\frac{A+B}{C+D}\le \max\{A/C,B/D\}$ for $A,B,C,D>0$ (and the generalization
to sums with more than two terms in the enumerator and in the denominator)
and consider parts of the fraction as follows: We form one group consisting of a red arc and $k$,
$|\Gamma_k|+|\Gamma_k^S|$ groups consisting of a (possibly speial) green arc and two red arcs,
and $\lceil z/2\rceil$ groups consisting of a blue arc and two red arcs (the last group may contain a single red arc). The ratios for these parts can be analysed as follows:
\begin{itemize}
\item For $k$ together with an arbitrary red arc $jk$, the ratio is
bounded by $\frac{(3+\nu)t_k}{(2+\nu)t_k}=1+\frac{1}{2+\nu}$,
as shown in Case~1.
\item For two red arcs together with one non-special green arc, the ratio is
bounded by 
$\frac{\frac{4}{\nu}+\frac{4}{\mu\nu}+\nu+\nu^2+1}{\frac{2}{\nu}+\frac{2}{\mu\nu}+\nu+\nu^2}$,
as shown in Case~1.
\item For two red arcs together with one special green arc, the ratio is
bounded by 
$\frac{\frac{4}{\nu}+\frac{4}{\nu\mu}+\nu+\frac{1}{\mu+1}}{\frac{2}{\nu}+\frac{2}{\mu\nu}+\nu}$
as shown in Case~1.
\item The ratio for a group consisting of red arcs $i_1k$ and $i_2k$ and
a (special or non-special) blue arc $ij$ is bounded by
$\frac{2D^*(i_1,k)+2D^*(i_2,k)+\rho^B D^*(i,j)}{
D^*(i_1,k)+D^*(i_2,k)+D^*(i,j)}$.
As $\rho^B < 2$, the ratio is maximized when
$D^*(i_1,k)$ and $D^*(i_2,k)$ are as large as possible and $D^*(i,j)$
is as small as possible.
By Lemma~\ref{lem:blue}, $D^*(i,j)\ge t_k$.
Furthermore, by Lemma~\ref{lem:rededgejsmaller}
we know $\sigma_i\le \sigma_k$. Hence,
$D^*(i_1,k)=t_{i_1}+p_{i_1}\le m_{i_1}(1+\frac{1}{\mu}) \le t_k(\frac{1}{\nu} + \frac{1}{\mu\nu})$,
and similarly $D^*(i_2,k)\le t_k(\frac{1}{\nu} + \frac{1}{\mu\nu})$.
Hence, we get:
\begin{align*}
\frac{2D^*(i_1,k)+2D^*(i_2,k)+\rho^B D^*(i,j)}{D^*(i_1,k)+D^*(i_2,k)+D^*(i,j)}
&\le
\frac{4 t_k(\frac{1}{\nu} + \frac{1}{\mu\nu}) + \rho^B t_k}{
2 t_k(\frac{1}{\nu} +\frac{1}{\mu\nu}) + t_k}\\
&=\frac{4(\frac{1}{\nu} + \frac{1}{\mu\nu})+\rho^B}{
2 (\frac{1}{\nu} +\frac{1}{\mu\nu})+1 }\\
&=\frac{4+\frac{5}{\mu}+\nu}{2+\frac{2}{\mu}+\nu}
\end{align*}
\end{itemize}
Hence, $\rho_{C_k}\le \max\{1+\frac{1}{2+\nu},
\frac{\frac{4}{\nu}+\frac{4}{\mu\nu}+\nu+\nu^2+1}{\frac{2}{\nu}+\frac{2}{\mu\nu}+\nu+\nu^2},
\frac{\frac{4}{\nu}+\frac{4}{\nu\mu}+\nu+\frac{1}{\mu+1}}{\frac{2}{\nu}+\frac{2}{\mu\nu}+\nu},
\frac{4+\frac{5}{\mu}+\nu}{2+\frac{2}{\mu}+\nu},
\}$.

It remains to show that Invariant~\ref{inv:full} is maintained.
Let $j'$ be an arbitrary vertex in $N^-(k)$.
For the case $j'=j$ (where $j$ is still the earliest vertex in
$N^-(k)$, based on the order of $V_I$), we have already shown that the invariant
holds. Therefore, assume that $j'\neq j$.
Let $P_{\ge j'}(k)$ be the set of the vertices in $P(k)$ that
have a red arc from~$j'$. Let $P_{<j'}(k)=P(k)-P_{\ge j'}(k)$.

Let $N^-_{\ge j'}(k)=\{s\in N^-(k)\mid m_s\ge m_{j'}\}$
and $N^-_{<j'}(k)=\{s\in N^-(k)\mid m_s< m_{j'}\}$ be a partition
of $N^-(k)$ into the elements before $j'$ and the elements
from $j'$ onward (in the order of $V_I$).
Let $\ell = |N^-_{<j'}(k)|$.
If $k'$, the vertex processed just before $k$, exists
and is in $P_{\ge j'}(k)$,
we apply Invariant~\ref{inv:full} to $k'$ and $j'$
to obtain
that the number of blue arcs between vertices in $\{s\in N^-(k')\mid m_s\ge m_{j'}\}$,
and therefore also between vertices in $N^-_{\ge j'}(k)$,
that have been used up in the analysis of vertices in $P_{\ge j'}(k)$
is at most $\sum_{r\in P_{\ge j'}(k)}(o_{\ge j'}(r)-1)$.
If $k'$ does not exist (i.e., if $k$ is the first vertex
to be processed) or is not in $P_{\ge j'}(k)$, then
no blue arcs between vertices in $N^-_{\ge j'}(k)$ have
been used up before processing~$k$.

If $|P(k)|\ge \ell/2$, then at least $1+2|\Gamma_k\cup\Gamma_k^S|\ge 1+ \ell$ of $k$'s red arcs
are grouped with $k$ or with green arcs, and the number
$z=o(k)-1-2|\Gamma_k\cup \Gamma_k^S|$ of red arcs that need to be grouped with blue arcs
is bounded by $z\le o(k)-1-\ell$. Therefore, the number of blue arcs
used in the analysis of $k$ is $\lceil z/2 \rceil \le z \le o(k)-\ell-1$.
Using $o_{\ge j'}(k)=o(k)-\ell$, we get that at most $o_{\ge j'}(k)-1$
blue arcs are used in the analysis of $k$, and hence
the total number of blue arcs between vertices in $N^-_{\ge j'}(k)$
that have been used just after $k$ has been processed is bounded by
$\sum_{r\in P_{\ge j'}(k)\cup \{k\} }(o_{\ge j'}(r)-1)$, and
Invariant~\ref{inv:full} holds.

Now, assume $|P(k)| < \ell/2$.
Note that $|P(k)|\le \frac{\ell}{2}-\frac12 = \frac{\ell-1}{2}$
as $|P(k)|$ and $\ell$ are integers.
In the analysis of vertices in $P_{<j'}(k)$, all the blue arcs
between vertices in $N^-(k)$ that have been used must have their
earlier endpoint in $N^-_{<j'}(k)$.
As observed earlier, by Invariant~\ref{inv:full} applied to $j$ just before $k$ is processed,
the number of blue arcs between vertices in $N^-(k)$ that
have been used in the analysis of vertices in $P(k)$
is at most
$\sum_{r\in P(k)} (o_{\ge j}(r)-1)$.
As $z=o(k)-1-2|\Gamma_k\cup\Gamma_k^S|=o(k)-1-2|P(k)|$,
we use $\lceil z/2\rceil =\left\lceil \frac{o(k)-1-2|P(k)|}{2} \right\rceil\le o(k)-1$
blue arcs in the analysis of~$k$.
The total number of blue arcs between vertices in $N^-(k)$ that
have been used after $k$ is processed is then at most
$$
o(k)-1 + \sum_{r\in P(k)} (o_{\ge j}(r)-1) \,.
$$
There are ${\ell \choose 2}=\frac{\ell(\ell-1)}{2}$
blue arcs between vertices in $N^-_{<j'}(k)$, and
at least $\ell$ further blue arcs with one endpoint
in $N^-_{<j'}(k)$ and one endpoint in $N^-_{\ge j'}(k)$.
This is a total of $\frac{\ell(\ell+1)}{2}$
blue arcs, all having their earlier endpoint
in $N^-_{<j'}(k)$.
As blue arcs with earlier endpoints are used first
in the analysis of a vertex, blue arcs with earlier
endpoint in $N^-_{\ge j'}(k)$ are only used (by vertices
with a red arc from $j$) if all blue arcs
with earlier endpoint in $N^-_{<j'}(k)$ have been used up.
Thus, the number of blue arcs between vertices in $N^-_{\ge j'}(k)$
that are used up just after $k$ has been processed
is at most
\begin{align*}
&o(k)-1 + \sum_{r\in P(k)} (o_{\ge j}(r)-1)
- \frac{\ell(\ell+1)}{2} \\
& \le \ell+(o(k)-\ell-1) + \sum_{r \in P_{<j'}(k)} (\ell-1)
+ \sum_{r \in P_{\ge j'}(k)} (\ell + (o_{\ge j}(r)-\ell-1))) 
- \frac{\ell(\ell+1)}{2} \\
& = \ell+(o(k)-\ell-1) + \sum_{r \in P_{<j'}(k)} (\ell-1)
+ \sum_{r \in P_{\ge j'}(k)} (\ell + (o_{\ge j}(r)-\ell-1)))
- \frac{\ell(\ell+1)}{2} \\
& = (|P(k)|+1)\ell + (o(k)-\ell-1)+ \sum_{r \in P_{\ge j'}(k)} (o_{\ge j}(r)-\ell-1)
- \frac{\ell(\ell+1)}{2} \\
& \le (\frac{\ell-1}{2} +1)\ell + (o_{\ge j'}(k)-1) + \sum_{r \in P_{\ge j'}(k)} (o_{\ge j'}(r)-1)
- \frac{\ell(\ell+1)}{2} \\
& = \frac{(\ell+1)\ell}{2} + \sum_{r \in P_{\ge j'}(k)\cup \{k\} } (o_{\ge j'}(r)1)
- \frac{\ell(\ell+1)}{2} \\
& = \sum_{r \in P_{\ge j'}(k) \cup \{k\}} (o_{\ge j'}(r)-1)
\end{align*}
Therefore, Invariant~\ref{inv:full} also holds for $k$ and $j'$.
As $j'$ was arbitrary, we have shown that Invariant~\ref{inv:full}
is maintained.
\end{proof}

By Lemma~\ref{lem:inv}, we know that for every $k\in V_R$ there
is a set $C_k$ of vertices and arcs such that $\rho_{C_k}\le \rho^C$.
Furthermore, the sets $C_k$ are pairwise disjoint; if an arc is
used as a blue arc in one set and as a green arc in another set,
it is split into a special blue arc and a special green arc,
and each set uses one of the two special arcs.
Let $V'$ denote the vertices that are not in any $C_k$ and note
that any vertex $j\in V'$ delays itself by $\sigma_j$ in both
the optimal schedule and the algorithm's schedule.
Let $A'$ denote the arcs that are not in any $C_k$ and note
that any arc $ij\in A'$ has
$D(i,j)\le \rho^N D^*(i,j)$ with
$\rho^N=\max\left\{ \frac{2\mu+1}{\mu+1}, 1+\nu \right\}$
by Lemma~\ref{lem:notredratio}.
We use $D(C_k)$ and $D^*(C_k)$ to denote the sum of the delays
in $C_k$ in the algorithm's schedule and in the optimal schedule,
respectively.
As the competitive ratio is
\begin{align*}
\frac{\alg}{\opt}
&= \frac{\sum_{j\in V'} \sigma_j + \sum_{ij\in A'} D(i,j) + \sum_{k\in V_R} D(C_k)}
{\sum_{j\in V'} \sigma_j + \sum_{ij\in A'} D^*(i,j) + \sum_{k\in V_R} D^*(C_k)}\\
& \le \frac{\sum_{j\in V'} \sigma_j + \sum_{ij\in A'} \rho^N D^*(i,j) + \sum_{k\in V_R} \rho^C D^*(C_k)}
{\sum_{j\in V'} \sigma_j + \sum_{ij\in A'} D^*(i,j) + \sum_{k\in V_R} D^*(C_k)}
\,,
\end{align*}
we get that the ratio is bounded by
$\max\{1,\rho^C,\rho^N\}=\max\{\rho^C,\rho^N\}$.
This completes the proof of Theorem~\ref{th:betaSort}.

\subsection{Lower bounds on the competitive ratio of $\beta$-SORT} \label{subsec:arbitrary:lower}
First, assume $\beta\le 1$.
Consider the following instance of the problem (where $M$ is a fixed positive number and $\epsilon>0$ is infinitesimally small):
$\gamma n$ \emph{short} jobs with $t_j=0,p_j=M$ and $(1-\gamma)n$ \emph{long} jobs
with $t_j=\frac{M}{\beta}-2\epsilon , p_j=M-\epsilon$.
The algorithm schedules the $\gamma n$ tests of short jobs, then the $(1-\gamma)n$ tests of long jobs,
then the processing operations of the long jobs, and finally the processing operations of the short jobs.
The optimal schedule schedules first all short jobs and then all long jobs.
To determine the limit of the ratio for large~$n$, we can assume $\epsilon=0$ and
omit terms linear in $n$ when we calculate the objective values of the optimal schedule
and the schedule produced by $\beta$-SORT. We have:
$$
\frac{\opt}{n^2} \approx M\gamma^2/2+M\gamma(1-\gamma)+(1+\frac{1}{\beta})M(1-\gamma)^2/2
$$
and
$$
\frac{\alg}{n^2} \approx (1-\gamma)\frac{M}{\beta} + M/2
$$
The ratio $\alg/\opt$ is maximized for $\gamma=(\beta+2-\sqrt{\beta(\beta+4)})/2$,
yielding ratio $\frac12 (\sqrt{\frac{\beta+4}{\beta}}+1)$. For $0<\beta\le 1$, this
function is decreasing in $\beta$ and obtains value $\phi=\frac15(\sqrt{5}+1)\approx 1.618$
for $\beta=1$.

Now, assume $\beta\ge 1$.
Consider the following instance of the problem (where $\epsilon>0$ is again infinitesimally small):
$\gamma n$ \emph{short} jobs with $t_j=1+2\epsilon,p_j=0$ and $(1-\gamma)n$ \emph{long} jobs
with $t_j=1, p_j=\beta+\epsilon$. The algorithm schedules the tests of the long jobs, then
the processing parts of the long jobs, then the short jobs (with each test followed immediately
by the execution of the tested job). The optimum schedule schedules first all short jobs and
then all long jobs.
We have (again with $\epsilon=0$ and omitting terms linear in~$n$):
$$
\frac{\opt}{n^2} = \gamma^2/2+\gamma(1-\gamma)+(1-\gamma)^2/2(1+\beta)
$$
and
$$
\frac{\alg}{n^2} = (1-\gamma)+\beta(1-\gamma)^2/2+\beta(1-\gamma)\gamma+ \gamma^2/2
$$
The ratio $\alg/\opt$ is maximized for
$\gamma=\frac{-1+2\beta+2\beta^2-\sqrt{1-4\beta+4\beta^2+4\beta^3}}{2\beta^2}$,
yielding ratio $(\sqrt{4\beta(\beta^2+\beta-1)+1}+1)/2\beta$. This function
equals $\phi$ for $\beta=1$ and increases with $\beta$ for $\beta>1$.

\begin{figure}
\centerline{\scalebox{0.8}{\includegraphics{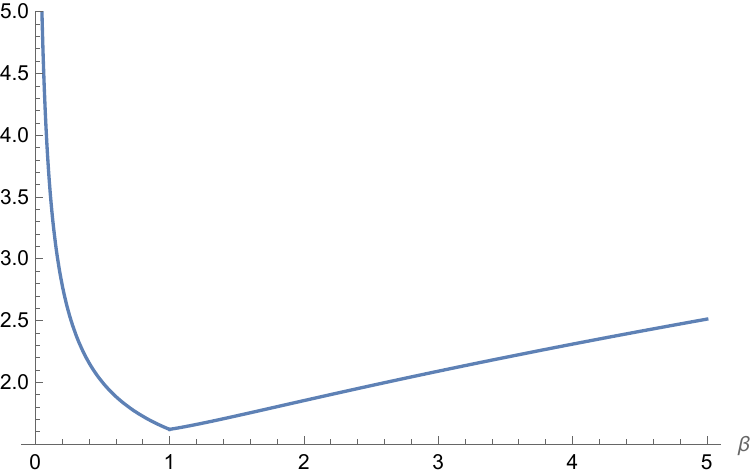}}}

\caption{Lower bound on the competitive ratio of $\beta$-SORT as a function of $\beta$.}
\label{fig:lbexampleratio}
\end{figure}%
This shows that $\beta$-SORT with $\beta=1$ is not better than 1.618-competitive,
and for every $\beta\neq 1$ we get a larger lower bound on the competitive ratio
of $\beta$-SORT. For example, we get lower bounds of $2$ for $\beta=0.5$,
$1.688$ for $\beta=4/3$, and $1.851$ for $\beta=2$. A plot of the lower bound
as a function of $\beta$ is shown in Fig.~\ref{fig:lbexampleratio}.

\section{Uniform Test Times}\label{sec:uniform}
In this section we assume that $t_j=1$ for all jobs~$j$. First,
we show the following lower bound.

\begin{theorem}
\label{th:obliglowerx0}
No deterministic algorithm can have competitive ratio strictly smaller than
$\sqrt{2}$ for the setting with obligatory tests and uniform test times.
\end{theorem}

\begin{proof}
Let an arbitrary deterministic algorithm for the problem
be given.
Consider the following adversarial construction, with a parameter
$\gamma$, $0\le\gamma\le 1$, whose value will be determined
later:
The adversary presents $n$ jobs. For the first $\gamma n$ jobs that are
tested by the algorithm, the adversary sets $p_j=1$. For the remaining
$(1-\gamma)n$ jobs, the adversary sets $p_j=0$. We call the jobs
with $p_j=0$ \emph{short} jobs and those with $p_j=1$ \emph{long} jobs.

The optimal schedule will schedule the jobs in SPT order, i.e., it will
first execute the $(1-\gamma)n$ short jobs and after that the $\gamma n$ long jobs,
always executing the processing part of a job right after its test.
The objective value of the optimal schedule can be written
as the sum of three parts:
\begin{align*}
P_1 & = \sum_{j=1}^{(1-\gamma)n} j = \frac{(1-\gamma)n((1-\gamma)n+1)}{2}\\
P_2 & = (1-\gamma) n \cdot \gamma n = \gamma(1-\gamma) n^2\\
P_3 & = \sum_{j=1}^{\gamma n} (2j) = \gamma n (\gamma n + 1)
\end{align*}
Here, $P_1$ is the sum of the completion times of the short jobs,
$P_2$ is the total delay added by the short jobs to the completion
times of the long jobs, and $P_3$ is the sum of the completion
times of the long jobs calculated as if their schedule started
at time~$0$. As $\opt=P_1+P_2+P_3$, we get:
$$
\opt = \frac{\gamma^2+1}{2} n^2 + \frac{1+\gamma}{2} n
$$
For the algorithm, we claim that it is best to schedule
the processing part of each job right after its test. 
For short jobs, this is obvious, because a short job
that was scheduled at a later time could be moved
forward to right after its test without affecting
the completion times of other jobs. This implies
in particular that no two jobs complete at the same
time, and that the completion times of any two
jobs are at least one time unit apart. Furthermore,
as it is clear that introducing idle time into
the schedule cannot help, we can assume that
all tests start and end at integral times and
that all job completion times are integers.
Now assume for a contradiction that some long job $j$
is the first job for which the
test completes at some time $\tau$ but the processing
part completes at some time $\tau+k$ for $k>1$.
This implies that no job completes at time
$\tau$, and $r \le k-1$ jobs have completion
times in the interval $[\tau,\tau+k-1]$.
Moving job $j$ forward to right after its test (and
shifting all tests and job executions from time
$\tau$ to $\tau+k-1$ one time unit later)
produces a schedule in which the completion
time of job $j$ decreases by $k-1$ while
the completion times of only $r\le k-1$ jobs
increase by~$1$. Therefore, the modified schedule
has an objective value that is the same or better.
By repeating this transformation, we obtain
a schedule where the processing part of each
job is executed right after its test, without
increasing the objective value.
Therefore, executing the processing part of
each job right after its
test is the best the algorithm can do.

The objective value of the schedule produced
by the agorithm is then
$\alg=P_3+P_2'+P_1$, where $P_2'=2\gamma n(1-\gamma) n=2\gamma(1-\gamma)n^2$
is the total delay that the long jobs with total length
$2\gamma n$ add to the completion times of the $(1-\gamma)n$
short jobs. Thus, the objective value $\alg$ of the
schedule produed by the algorithm is:
$$
\alg = \frac{1+2\gamma-\gamma^2}{2}n^2 + \frac{1+\gamma}{2} n
$$
The competitive ratio, as a function of $\gamma$, is then:
$$
\rho(\gamma) =\frac{\alg}{\opt} = 
\frac{
\frac{1+2\gamma-\gamma^2}{2}n^2 + \frac{1+\gamma}{2} n
}{
\frac{\gamma^2+1}{2} n^2 + \frac{1+\gamma}{2} n
}
= \frac{1+2\gamma-\gamma^2 + \frac{1+\gamma}{n}}
{\gamma^2+1  + \frac{1+\gamma}{n}}
$$
For fixed $\gamma$, $\rho(\gamma)$ increases with $n$, and the ratio
converges to $\frac{1+2\gamma-\gamma^2}{\gamma^2+1}$ for $n\to\infty$.
The function $f(\gamma)=\frac{1+2\gamma-\gamma^2}{\gamma^2+1}$
has a global maximum at $\gamma_0=\sqrt{2}-1$ with $f(\gamma_0)=\sqrt{2}$,
as can be shown using standard methods from calculus.
Thus, if the adversary presents instances with arbitrarily large $n$
and sets $\gamma$ to the multiple of $\frac{1}{n}$ closest to
$\sqrt{2}-1$, it can force the algorithm to have competitive
ratio arbitrarily close to $\sqrt{2}$.
\end{proof}

Now, we consider algorithms for the problem
with uniform test times.
The simple algorithm that first tests all jobs
and, after all tests have been completed executes the
processing parts of all jobs in SPT order can
be shown to be $2$-competitive: If $\opt_P$ denotes
the sum of completion times of scheduling the processing
parts (ignoring the tests) of the $n$ given jobs in SPT order,
then the optimal sum of completion times
is $\frac{n(n+1)}{2}+\opt_P$ while the algorithm has
sum of completion times $n^2+\opt_P$.

To beat the competitive ratio of~$2$ (and the ratio $1.861$
that we have proved for $1$-SORT), we
propose algorithm SIDLE (Short Immediate, Delay Long Executions)
that has a parameter~$y>0$. The algorithm
tests all jobs, and executes a job $j$ immediately after its
test if $p_j\le y$ (\emph{short} job). The jobs $j$ with $p_j>y$ (\emph{long} jobs) are executed
in SPT order at the end of the schedule, after all jobs have been
tested and all short jobs executed. The algorithm is inspired
by algorithm \textsc{THRESHOLD} from~\cite{DurrEMM20}.

\begin{theorem}
\label{th:sidle}
Algorithm SIDLE with $y=y_0 \approx 1.35542$ has competitive
ratio at most
$\frac{1}{2}(1-y_0+y_0^2+\sqrt{9-2y_0-y_0^2-2y_0^3+y_0^4})
\approx 1.58451\le 1.585$.
Here, $y_0$ is the second root of
the polynomial $2y^3-9y^2+10y-2$.
\end{theorem}

\begin{proof}
Assume that there are $k=\alpha n$ short jobs (i.e., jobs $j$ with $p_j\le y$) and $n-k
= (1-\alpha)n$ long jobs (i.e., jobs $j$ with $p_j>y$).

Let $P_S$ be the sum of the processing times of all short
jobs (not including their test times).

Let $\opt_S$ denote the cost of an SPT-schedule for the
short jobs (without their tests), and let $\opt_L$ denote the
cost of an SPT-schedule for the long jobs (without their tests).
Let $\opt$ be the cost of the optimal
schedule for all jobs. We have:
$$
\opt = \frac{k(k+1)}{2}+\opt_S + (P_S+k) \cdot (n-k) + \frac{(n-k)(n-k+1)}{2} + \opt_L
$$
Here, the term $\frac{k(k+1)}{2}+\opt_S$ is the cost of an SPT-schedule
for the short jobs (including their tests), $(P_S+k)\cdot (n-k)$ is the total delay that
the short jobs cause for the $n-k$ long jobs, and $\frac{(n-k)(n-k+1)}{2} + \opt_L$
is the cost of an SPT schedule for the long jobs (including their tests),
cf.\ Equation~(\ref{eq:uniformsum}) in Section~\ref{sec:prelim}.

It is clear that the worst order in which algorithm SIDLE could test the
jobs is: First the $n-k$ tests of long jobs in arbitrary order, then the
$k$ tests of short jobs in reverse SPT order.
We split the schedule produced by SIDLE for this order into the
following three parts:
\begin{itemize}
\item $n-k$ tests of long jobs. They increase the cost of the remaining
two parts by $(n-k)n$.
\item $k$ tests of short jobs, each followed immediately by the
execution of the short job that was just tested.
Denote the sum of the completion times of the short jobs without their
tests in reverse SPT order by $\alg_S$. Then the sum of completion
times within this block (calculated as if it were to start at time~$0$) is
$\frac{k(k+1)}{2}+\alg_S$.
This part increases the cost of the third part by $(P_S+k)(n-k)$.
\item Execution of the long jobs in SPT order. The cost for this
part alone (if it were to start at time~$0$) is $\opt_L$.
\end{itemize}
The total cost of the solution by SIDLE can then be written
as:
\begin{eqnarray*}
\alg & = & (n-k)n + \frac{k(k+1)}{2}+ \alg_S + (P_S+k)(n-k) + \opt_L \\
    & = & (n-k)(n+k) + \frac{k(k+1)}{2} + \alg_S + P_S\cdot (n-k) + \opt_L\\
    & = & n^2-k^2 + \frac{k(k+1)}{2} + \alg_S + P_S\cdot(n-k)+\opt_L
\end{eqnarray*}
We want to bound this ratio:
\begin{equation}
\frac{\alg}{\opt} =
\frac{n^2-k^2 + \frac{k(k+1)}{2} + \alg_S + P_S\cdot(n-k)+\opt_L}{%
\frac{k(k+1)}{2}+\opt_S + (P_S+k) \cdot (n-k) + \frac{(n-k)(n-k+1)}{2}+\opt_L
}
\label{eq:sidleratio1}
\end{equation}
The following lemma from \cite{DurrEMM20} allows us to restrict our
attention to the case where the processing times of the short jobs are
taken from a very limited set of possibilities.

\begin{lemma}[Lemma~3 in \cite{DurrEMM20}]\label{lem:Durr}
    Suppose that there is an interval $[\ell',u']$ such that $\opt$ schedules all jobs~$j$ with $p_j\in[\ell',u']$ either all tested or all untested, independently of the actual processing time in $[\ell',u']$.  Suppose that this holds also for $\alg$.
Moreover, suppose that the schedules of $\opt$ and $\alg$ do not change
(in the sense that the order of all tests and job executions remains the same)
when changing the processing times in $[\ell',u']$ as long as the
relative ordering of job processing times does not change.
    Then there is a worst-case instance for $\alg$ where every job $j$ with $p_j\in[\ell',u']$ satisfies $p_j\in\{\ell',u'\}$.
\end{lemma}

Applying Lemma~\ref{lem:Durr} with $\ell'=0$ an $u'=y$, we conclude that
it suffices to consider the
case that all the processing times of the short jobs are
equal to $0$ or~$y$.
Let $\gamma$ be such that $\gamma k$ short jobs
have processing time $0$ and $(1-\gamma)k$ short jobs have
processing time~$y$.

Then we have
\begin{align*}
\opt_S &= \frac{(1-\gamma)k((1-\gamma)k+1)}{2} y = \frac{(1-\gamma)^2}{2}y k^2 + \frac{1-\gamma}{2}y k
\end{align*}
and, because the $\gamma k$ short jobs complete at time $(1-\gamma)k y$ instead of time $0$ in the schedule in reverse SPT order,
\begin{align*}
\alg_S &= \opt_S + \gamma k (1-\gamma) k y = \frac{1-\gamma^2}{2} y k^2 +  \frac{1-\gamma}{2}y k
\,.
\end{align*}
Substituting these expressions as
well as $P_S = (1-\gamma)k y$ into (\ref{eq:sidleratio1}) gives the following
expression for $\alg/\opt$:
\begin{align}
& 
\frac{n^2-k^2 + \frac{k(k+1)}{2} + \frac{1-\gamma^2}{2} y k^2 +  \frac{1-\gamma}{2}y k + (1-\gamma)k y \cdot(n-k)+\opt_L}{%
\frac{k(k+1)}{2}+\frac{(1-\gamma)^2}{2}y k^2 + \frac{1-\gamma}{2}y k + ((1-\gamma)k y +k) \cdot (n-k) + \frac{(n-k)(n-k+1)}{2}+\opt_L
}\nonumber \\
&
\le 
\frac{n^2-k^2 + \frac{k(k+1)}{2} + \frac{1-\gamma^2}{2} y k^2 +  \frac{1-\gamma}{2}y k + (1-\gamma)k y \cdot(n-k)+\frac{(n-k)(n-k+1)}{2}y}{%
\frac{k(k+1)}{2}+\frac{(1-\gamma)^2}{2}y k^2 + \frac{1-\gamma}{2}y k + ((1-\gamma)k y +k) \cdot (n-k) + \frac{(n-k)(n-k+1)}{2}(1+y)
}\nonumber \\
&
=
\frac{
(1-\alpha^2+\frac{\alpha^2}{2}+\frac{1-\gamma^2}{2}y\alpha^2+
(1-\gamma)\alpha y (1-\alpha)+\frac{(1-\alpha)^2}{2}y)n^2
+(\frac{\alpha}{2}+\frac{1-\gamma}{2}y\alpha+\frac{1-\alpha}{2}y)n
}{
(\frac{\alpha^2}{2}+\frac{(1-\gamma)^2}{2}y\alpha^2+
((1-\gamma)y+1)\alpha(1-\alpha) + \frac{(1-\alpha)^2}{2}(1+y))n^2
+(\frac{\alpha}{2}+\frac{1-\gamma}{2}y\alpha+\frac{1-\alpha}{2}(1+y))n
}\nonumber \\
&
\le
\frac{
(1-\alpha^2+\frac{\alpha^2}{2}+\frac{1-\gamma^2}{2}y\alpha^2+
(1-\gamma)\alpha y (1-\alpha)+\frac{(1-\alpha)^2}{2}y)n^2
}{
(\frac{\alpha^2}{2}+\frac{(1-\gamma)^2}{2}y\alpha^2+
((1-\gamma)y+1)\alpha(1-\alpha) + \frac{(1-\alpha)^2}{2}(1+y))n^2
}\nonumber \\
&
=
\frac{
1-\frac{\alpha^2}{2} + \frac{y}{2} (1+2\alpha^2\gamma-\alpha^2\gamma^2-2\alpha\gamma)
}{
\frac12 + \frac{y}{2}(1+\alpha^2\gamma^2-2\alpha\gamma)
}
\label{eq:sidleratio2}
\end{align}
Here, the first inequality holds because the ratio is maximized when $\opt_L$ is as small as possible, which happens when all the
long jobs have processing time $y+\epsilon$ for infinitesimally
small $\epsilon$, giving $\opt_L=\frac{(n-k)(n-k+1)}{2}y$.
The second inequality holds because the term linear in $n$
in the enumerator is smaller than the term linear in $n$
in the denominator, and so omitting these terms cannot
decrease the ratio.

\begin{figure}
\centerline{\scalebox{0.8}{\includegraphics{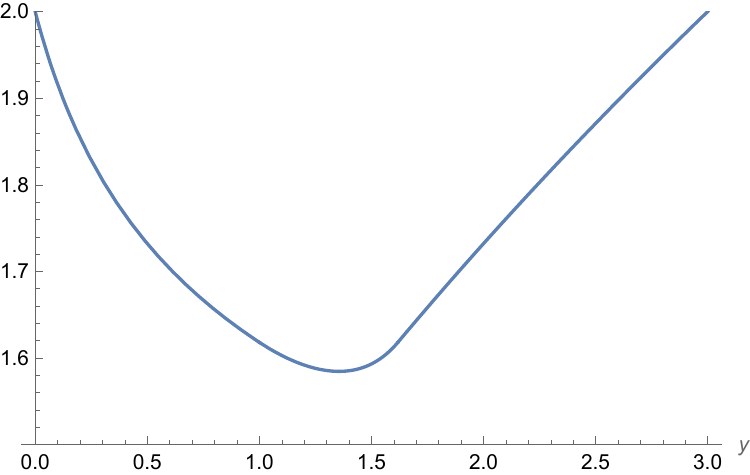}}}

\caption{Competitive ratio of algorithm SIDLE as a function
of the parameter~$y$.}
\label{fig:sidleratio}
\end{figure}%
For any fixed~$y$, the adversary can choose $\alpha$ and $\gamma$
with $0\le \alpha\le 1$ and $0\le \gamma\le 1$ so as
to maximize (\ref{eq:sidleratio2}). Let $\rho(y)$ denote
the value of that maximum. A plot of the function $\rho(y)$,
determined using Mathematica, is shown
in Figure~\ref{fig:sidleratio}. We want to choose
$y\ge 0$ in such a way that $\rho(y)$ is minimized.

Using Mathematica, we find the choice
of $y$ that minimizes $\rho(y)$ is $y=y_0\approx 1.35542$,
where $y_0$ is the second root of the polynomial
$2y^3-9y^2+10y-2$. The resulting value of $\rho(y)$
is $\frac{1}{2}(1-y_0+y_0^2+\sqrt{9-2y_0-y_0^2-2y_0^3+y_0^4})
\approx 1.58451$.
\end{proof}

We remark that the analysis of Theorem~\ref{th:sidle}
is tight: For $\alpha\approx 0.644584$ and
$\gamma\approx 0.737781$ (the values that maximize
(\ref{eq:sidleratio2}) for $y=y_0$), we can consider instances
with $\alpha\gamma n$ jobs with processing time~$0$,
$\alpha(1-\gamma)n$ jobs with processing time~$y_0$, and
$(1-\alpha)n$ jobs with processing time $y_0+\epsilon$
for infinitesimally small~$\epsilon$. For large enough~$n$,
the competitive ratio of algorithm SIDLE on these instances
is then approximately~$1.58451$.

\section{Conclusion}
\label{sec:conc}%
In this paper, we have introduced a variant of scheduling with
testing where every job must be tested and the objective is
minimizing the sum of completion times. Our main result is
an analysis showing that the competitive analysis of the
$1$-SORT algorithm is at most $1.861$. For the special case of uniform
test times, we have presented a $1.585$-competitive algorithm
as well as a lower bound of $\sqrt{2}$ on the competitive
ratio of any deterministic algorithm.

There are several interesting directions for future research.
First, there are gaps between our lower bound of
$\sqrt{2}$ and our upper bounds of $1.585$ and $1.861$
on the competitive ratio for uniform and arbitrary test times,
respectively.
One immediate question is whether our analysis of $1$-SORT can be
improved, as we only know that the competitive
ratio of $1$-SORT is not better than $1.618$.
Our lower
bound of $\sqrt{2}$ on the competitive ratio of deterministic
algorithms holds for uniform test times; it would
be interesting to find out whether the case of arbitrary test times
admits a stronger lower bound. Furthermore, it would be worthwhile
to study randomized algorithms for the problem. Finally,
it would be interesting to explore whether our new technique for
analyzing $\beta$-SORT can also be applied to other variants
of scheduling with testing.

\bibliography{testing}

\appendix
\section{Appendix: Mathematica notebooks}
On the following pages we provide the Mathematica notebooks
showing how the ratio of Theorem \ref{th:betaSort} has been minimized
to obtain Corollary~\ref{cor:betaSort},
how the competitive ratio of $1$-SORT on the examples presented
in Section \ref{subsec:arbitrary:lower} has been maximized, and how the
upper bound on the competitive
ratio of SIDLE in Theorem~\ref{th:sidle} has been determined.



\section*{Supplementary material}

# Determining m (mu) and n (nu) to minimize the ratio in Theorem 1

In[27]:= `Clear[r, m, n, S, f1, f2, f3, f4, f5, f6, f7, f8, f9];`

The ratio is the maximum of the following 9 functions of m and n:

In[28]:= `{f1, f2, f3, f4, f5, f6, f7, f8, f9 } = {(n + n^2 + 2 + 2 / m) / (n + n^2 + 1 + 1 / m),`
`1 + 1 / (2 + n), (4 / n + 4 / (m * n) + n + n^2 + 1) / (2 / n + 2 / (m * n) + n + n^2),`
`(4 / n + 4 / (m * n) + n + 1 / (m + 1)) / (2 / n + 2 / (m * n) + n), (4 + 5 / m + n) / (2 + 2 / m + n),`
`1 + 1 / (n + n^2), 1 + 1 / (n (m + 1)), (2 m + 1) / (m + 1), 1 + n}`

$$\left\{\frac{2 + \frac{2}{m} + n + n^2}{1 + \frac{1}{m} + n + n^2}, \ 1 + \frac{1}{2 + n}, \ \frac{1 + \frac{4}{n} + \frac{4}{mn} + n + n^2}{\frac{2}{n} + \frac{2}{mn} + n + n^2},\right.$$

$$\left. \frac{\frac{1}{1+m} + \frac{4}{n} + \frac{4}{mn} + n}{\frac{2}{n} + \frac{2}{mn} + n}, \ \frac{4 + \frac{5}{m} + n}{2 + \frac{2}{m} + n}, \ 1 + \frac{1}{n + n^2}, \ 1 + \frac{1}{(1 + m)\, n}, \ \frac{1 + 2m}{1 + m}, \ 1 + n\right\}$$

Out[28]=

$$\left\{\frac{2 + \frac{2}{m} + n + n^2}{1 + \frac{1}{m} + n + n^2}, \ 1 + \frac{1}{2 + n}, \ \frac{1 + \frac{4}{n} + \frac{4}{mn} + n + n^2}{\frac{2}{n} + \frac{2}{mn} + n + n^2},\right.$$

$$\left. \frac{\frac{1}{1+m} + \frac{4}{n} + \frac{4}{mn} + n}{\frac{2}{n} + \frac{2}{mn} + n}, \ \frac{4 + \frac{5}{m} + n}{2 + \frac{2}{m} + n}, \ 1 + \frac{1}{n + n^2}, \ 1 + \frac{1}{(1 + m)\, n}, \ \frac{1 + 2m}{1 + m}, \ 1 + n\right\}$$

Out[29]=

$$\left\{\frac{2 + \frac{2}{m} + n + n^2}{1 + \frac{1}{m} + n + n^2}, \ 1 + \frac{1}{2 + n}, \ \frac{1 + \frac{4}{n} + \frac{4}{mn} + n + n^2}{\frac{2}{n} + \frac{2}{mn} + n + n^2},\right.$$

$$\left. \frac{\frac{1}{1+m} + \frac{4}{n} + \frac{4}{mn} + n}{\frac{2}{n} + \frac{2}{mn} + n}, \ \frac{4 + \frac{5}{m} + n}{2 + \frac{2}{m} + n}, \ 1 + \frac{1}{n + n^2}, \ 1 + \frac{1}{(1 + m)\, n}, \ \frac{1 + 2m}{1 + m}, \ 1 + n\right\}$$

We first determine the values of m and n that minimize the maximum of f3, f8 and f9:

In[30]:= `r = Max[f3, f8, f9]`

Out[30]=

$$\text{Max}\left[\frac{1 + 2m}{1 + m}, \ 1 + n, \ \frac{1 + \frac{4}{n} + \frac{4}{mn} + n + n^2}{\frac{2}{n} + \frac{2}{mn} + n + n^2}\right]$$

In[31]:= `S = Minimize[{r, m ≥ 1 && 0 ≤ n ≤ 1 && m ≥ 1 / n && 1 + 1 / m ≤ n + n^2}, {m, n}]`

Out[31]=

$$\left\{\frac{1 + 2 \boxed{\circ\ 6.16\ldots}}{1 + \boxed{\circ\ 6.16\ldots}}, \ \left\{m \to \boxed{\circ\ 6.16\ldots}, \ n \to \boxed{\circ\ 0.860\ldots}\right\}\right\}$$

The numerical value of the minimum is:

In[10]:= **N[S〚1〛]**

Out[10]=

   1.86039

The numerical values of m and n that produce the minimum are:

In[11]:= **N[S〚2〛]**

Out[11]=

   {m → 6.16277, n → 0.860389}

Finally, we verify that for the obtained values of m and n, the functions f1,f2,f4,f5,f6,f7 are all less than or equal to Max[f3,f8,f9]

In[18]:= **FullSimplify[(f1 /. S〚2〛) ≤ S〚1〛]**

Out[18]=

   True

In[19]:= **FullSimplify[(f2 /. S〚2〛) ≤ S〚1〛]**

Out[19]=

   True

In[21]:= **FullSimplify[(f4 /. S〚2〛) ≤ S〚1〛]**

Out[21]=

   True

In[22]:= **FullSimplify[(f5 /. S〚2〛) ≤ S〚1〛]**

Out[22]=

   True

In[23]:= **FullSimplify[(f6 /. S〚2〛) ≤ S〚1〛]**

Out[23]=

   True

In[24]:= **FullSimplify[(f7 /. S〚2〛) ≤ S〚1〛]**

Out[24]=

   True

# Analysis of Lower Bound examples for $\beta$ - SORT

Choosing g (gamma) to maximize the ratio for $\beta \leq 1$

In[26]:= **S1 = FullSimplify[Maximize[{((1 - g) / b + 0.5) /**
**(g^2 / 2 + g (1 - g) + (1 + 1 / b) (1 - g) ^2 / 2), 0 ≤ g ≤ 1 && 0 < b ≤ 1}, {g}]]**

Out[26]=
$$\left\{ \left[ \begin{array}{ll} 0.5 + 0.5 \sqrt{\frac{4. + b}{b}} & 0 < b \leq 1 \\ -\infty & \text{True} \end{array} \right., \left\{ g \to \left[ \begin{array}{ll} 1. + 0.5 b - 0.5 \sqrt{b (4. + b)} & 0 < b \leq 1 \\ \text{Indeterminate} & \text{True} \end{array} \right] \right\} \right\}$$

Choosing g (gamma) to maximize the ratio for $\beta \geq 1$

In[27]:= **S2 = FullSimplify[Maximize[{((1 - g) + b (1 - g) ^2 / 2 + b (1 - g) g + g^2 / 2) /**
**(g^2 / 2 + g (1 - g) + (1 - g) ^2 (1 + b) / 2), 0 ≤ g ≤ 1 && b ≥ 1}, {g}]]**

Out[27]=
$$\left\{ \left[ \begin{array}{ll} \frac{1 + \sqrt{1 + 4 b (-1 + b + b^2)}}{2 b} & b \geq 1 \\ -\infty & \text{True} \end{array} \right., \left\{ g \to \left[ \begin{array}{ll} -\frac{1 - 2 b (1 + b) + \sqrt{1 + 4 b (-1 + b + b^2)}}{2 b^2} & b \geq 1 \\ \text{Indeterminate} & \text{True} \end{array} \right] \right\} \right\}$$

Plot the ratio as a function of b ($\beta$)

In[40]:= **Plot[Piecewise[{{S1[[1]], 0 ≤ b ≤ 1}, {S2[[1]], b ≥ 1}}],**
**{b, 0, 5}, AxesLabel → {$\beta$}, PlotRange → {1.5, 5}]**

Out[40]=

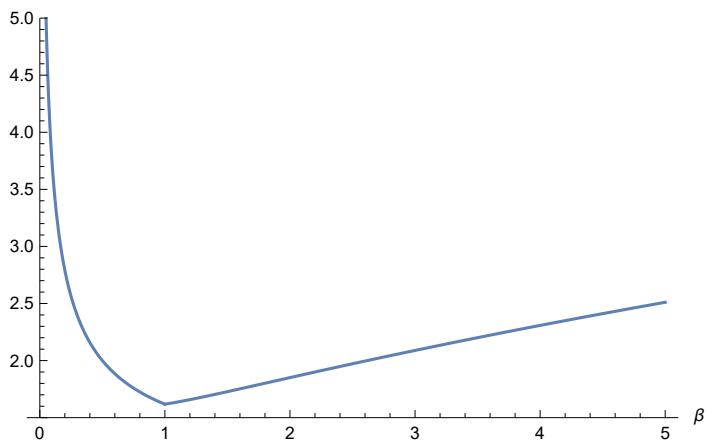

Zoom in on the range around $\beta$=1

In[38]:= `Plot[Piecewise[{{S1[[1]], 0 ≤ b ≤ 1}, {S2[[1]], b ≥ 1}}],`
`{b, 0.9, 1.1}, AxesLabel → Automatic, PlotRange → {1.55, 1.65}]`

Out[38]=

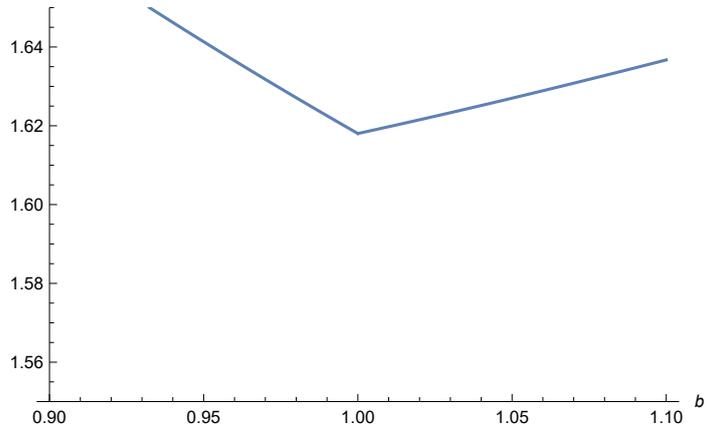

# Competitive ratio of algorithm SIDLE

In[9]:= `Clear[r, a, g, y, S, ysol, M]; r =`
    `(1 - a * a / 2 + y / 2 * (1 + 2 a^2 g - a^2 g^2 - 2 a g)) / (1 / 2 + y / 2 * (1 + a^2 g^2 - 2 * a * g))`

In[10]:=
$$\frac{1 - \frac{a^2}{2} + \frac{1}{2}\left(1 - 2\,a\,g + 2\,a^2\,g - a^2\,g^2\right)y}{\frac{1}{2} + \frac{1}{2}\left(1 - 2\,a\,g + a^2\,g^2\right)y}$$

Out[10]=

$$\frac{1 - \frac{a^2}{2} + \frac{1}{2}\left(1 - 2\,a\,g + 2\,a^2\,g - a^2\,g^2\right)y}{\frac{1}{2} + \frac{1}{2}\left(1 - 2\,a\,g + a^2\,g^2\right)y}$$

In[11]:= `S = Maximize[{r, 0 ≤ a ≤ 1 && 0 ≤ g ≤ 1 && y ≥ 0}, {a, g}]`

Out[11]=

$$\left\{\begin{cases} 2 & y == 0 \\ \sqrt{1+y} & y > \frac{1}{2}\left(1+\sqrt{5}\right) \\ \frac{-1+2\,y+y\,\sqrt{\frac{1+4y}{y^2}}}{2\,y} & 0 < y \le 1 \\ \frac{1}{2}\left(1 - y + y^2 + \sqrt{9 - 2\,y - y^2 - 2\,y^3 + y^4}\right) & y == \frac{1}{2}\left(1+\sqrt{5}\right)\ ||\ 1 < y < \frac{1}{2}\left(1+\sqrt{5}\right) \\ -\infty & \text{True} \end{cases}\right.,$$

$$\frac{1}{2}$$
$$\frac{\sqrt{1+y}}{1+\sqrt{1+y}}$$

$$-\frac{1}{\left(1 - y\,\sqrt{\frac{1+4y}{y^2}}\right)\left(-1+4\,y+y\,\sqrt{\frac{1+4y}{y^2}}\right)}$$

$$2\left(y + y^2\,\sqrt{\frac{1+4y}{y^2}} - \sqrt{2}\right.$$

$$2\,y^2\,(1+4$$

$$4\,y^4\,\sqrt{\frac{1+4}{y^2}}$$

$$\left(y\left(1 - y\,\sqrt{\phantom{a}}\right.\right.$$

$$\sqrt{2}\,y\,\sqrt{\frac{1+4y}{y^2}}\,\sqrt{-\left(\phantom{a}\right.}$$

$$2\,y^2\,(1+4$$

$$4\,y^4\,\sqrt{\frac{1+4}{y^2}}$$

$$\left(y\left(1 - y\,\sqrt{\phantom{a}}\right.\right.$$

$$2\,\sqrt{2}\,y^2\,\sqrt{\frac{1+4y}{y^2}}\,\sqrt{\left(\phantom{a}\right.}$$

$$2 y^2 \left(1 + \mathrel{?}\right.$$

$$4 y^4 \sqrt{\frac{1+4}{y^2}}$$

$$\left(y \left(1 - y \sqrt{\mathrel{?}}\right.\right.$$

$$\frac{1}{3 - y + y^2 + \sqrt{9 - 2 y - y^2 - 2 y^3 + y^4}} \left(1 - \right.$$

$$3 \sqrt{2} \sqrt{-\frac{3 - 4 y + y^2 - 2 y^3 + y^4}{\mathrel{?}}}$$

$$\sqrt{2} \, y \sqrt{-\frac{3 - 4 y + y^2 - 2 y^3 + y^4}{\mathrel{?}}}$$

$$\sqrt{2} \, y^2 \sqrt{-\frac{3 - 4 y + y^2 - 2 y^3 + y}{\mathrel{?}}}$$

$$\sqrt{2} \sqrt{9 - 2 y - y^2 - 2 y^3}$$

$$\sqrt{-\frac{3 - 4 y + y^2 - 2 y^3 + y^4 + \sqrt{9 - ?}}{\mathrel{?}}}$$

$$\left\{ a \to \begin{cases} 1 & y == \frac{1}{2}\left(1 + \sqrt{5}\right) \;|\;|\; y > \frac{1}{2}\left(1 + \sqrt{5}\right) \\ \dfrac{-1 + y \sqrt{\frac{1+4 y}{y^2}}}{1 + y \sqrt{\frac{1+4 y}{y^2}}} & 0 < y \le 1 \\ \dfrac{-y - y^2 + y^3 + y \sqrt{9 - 2 y - y^2 - 2 y^3 + y^4}}{3 - 3 y + y^2 + \sqrt{9 - 2 y - y^2 - 2 y^3 + y^4}} & 1 < y < \frac{1}{2}\left(1 + \sqrt{5}\right) \\ 0 & y == 0 \\ \text{Indeterminate} & \text{True} \end{cases} \quad,\; g \to \right.$$

$$\frac{1}{-y - y^2 + y^3 + y \sqrt{9 - 2 y - y^2 - 2 y^3 + y^4}}$$

$$\left(-1 - y + y^2 + \sqrt{9 - 2 y - y}\right.$$

$$17 y^4 - y^5 - ?$$

$$7 y^2 \sqrt{9 - 2 y}$$

$$3 y^4 \sqrt{9 - 2 y}$$

$$y^6 \sqrt{9 - 2 y - }$$

$$2 y^2 \left(9 - 2 y\right.$$

$$y \left(9 - 2 y - y\right.$$

$$\left(y \left(3 - y + y^2 + \right.\right.$$

$$\left(-y - y^2 + y^3\right.$$

$$\sqrt{2} \, y^2 \sqrt{-\left(\left(-27 y - ?\right.\right.}$$

$$9 y \sqrt{9 - 2 y}$$

$$3 y^3 \sqrt{9 - 2 y}$$

$$y^5 \sqrt{9 - 2 y - }$$

$$3 y \left(9 - 2 y - \right.$$

$$y^4 \left(9 - 2 y - \right.$$

$$y^2 \left(9 - 2 y - \right.$$

$$\left(y \left(3 - y + y^2 + \right.\right.$$

$$\left(-y - y^2 + y^3\right.$$

$$\sqrt{2} \, y^3 \sqrt{-\left(\left(-27 y - ?\right.\right.}$$

$$9 y \sqrt{9 - 2 y}$$

$$3 y^3 \sqrt{9 - 2 y}$$

$$y^5 \sqrt{9 - 2 y}$$

$$3 y \left(9 - 2 y\right.$$

$$y^4 \left(9 - 2 y\right.$$

$$y^2 \left(9 - 2 y\right.$$

$$\left(y \left(3 - y + y^2 +\right.\right.$$

$$\left(-y - y^2 + y^3\right.$$

$$\sqrt{2} \ y \ \sqrt{9 - 2 y - y^2 - 2}$$

$$y^5 - 3 y^6 - 2$$

$$7 y^2 \sqrt{9 - 2 y}$$

$$3 y^4 \sqrt{9 - 2 y}$$

$$y^6 \sqrt{9 - 2 y}$$

$$2 y^2 \left(9 - 2 y\right.$$

$$y \left(9 - 2 y - y\right.$$

$$\left(y \left(3 - y + y^2 +\right.\right.$$

$$\left(-y - y^2 + y^3\right.$$

`Indeterminate`

Plot of the competitive ratio as a function of the parameter y

In[12]:= `Plot[S〚1〛, {y, 0, 3}, AxesLabel → Automatic, PlotRange → {1.5, 2}]`

Out[12]=



Determining the value of y that minimizes the competitive ratio

In[13]:= `M = Minimize[{`$\frac{1}{2}$ `(1 - y + y`$^2$ `+ `$\sqrt{(9 - 2y - y^2 - 2y^3 + y^4)}$`), y ≥ 0}, {y}]`

Out[13]=

$\left\{ \frac{1}{2} \left( 1 - \boxed{⊘ \; 1.36\ldots} + \boxed{⊘ \; 1.36\ldots}^2 + \sqrt{9 - 2 \boxed{⊘ \; 1.36\ldots} - \boxed{⊘ \; 1.36\ldots}^2 - 2 \boxed{⊘ \; 1.36\ldots}^3 + \boxed{⊘ \; 1.36\ldots}^4} \right), \right.$
$\left. \left\{ y \to \boxed{⊘ \; 1.36\ldots} \right\} \right\}$

In[14]:= `N[M, 10]`

Out[14]=

`{1.584511115, {y → 1.355415727}}`

Recompute the optimal choice of y as the second root of the polynomial returned above

In[15]:= `ysol = Solve[-2 + 10 y - 9 y^2 + 2 y^3 == 0, y]`

Out[15]=

$\left\{ \left\{ y \to \boxed{⊘ \; 0.255\ldots} \right\}, \left\{ y \to \boxed{⊘ \; 1.36\ldots} \right\}, \left\{ y \to \boxed{⊘ \; 2.89\ldots} \right\} \right\}$

In[16]:= `N[ysol]`

Out[16]=

`{{y → 0.255356}, {y → 1.35542}, {y → 2.88923}}`

Substitute the optimal value of y into the formula for the ratio

In[17]:= `r = r /. ysol〚2〛`

Out[17]=

$$\frac{1 - \frac{a^2}{2} + \frac{1}{2} \left( 1 - 2\,a\,g + 2\,a^2\,g - a^2\,g^2 \right) \boxed{⊘ \; 1.36\ldots}}{\frac{1}{2} + \frac{1}{2} \left( 1 - 2\,a\,g + a^2\,g^2 \right) \boxed{⊘ \; 1.36\ldots}}$$

Determine the values of a (alpha) and g (gamma) that yield the worst case for that optimal value of y

In[18]:= `Maximize[{r, 0 ≤ a ≤ 1 && 0 ≤ g ≤ 1}, {a, g}]`

Out[18]=

$\left\{ \left( 1 - \frac{1}{2} \boxed{⊘ \; 0.645\ldots}^2 + \right. \right.$
$\frac{1}{2} \left( 1 - 2 \boxed{⊘ \; 0.645\ldots} \boxed{⊘ \; 0.738\ldots} + 2 \boxed{⊘ \; 0.645\ldots}^2 \boxed{⊘ \; 0.738\ldots} - \boxed{⊘ \; 0.645\ldots}^2 \boxed{⊘ \; 0.738\ldots}^2 \right)$
$\left. \boxed{⊘ \; 1.36\ldots} \right) \Big/$
$\left( \frac{1}{2} + \frac{1}{2} \left( 1 - 2 \boxed{⊘ \; 0.645\ldots} \boxed{⊘ \; 0.738\ldots} + \boxed{⊘ \; 0.645\ldots}^2 \boxed{⊘ \; 0.738\ldots}^2 \right) \boxed{⊘ \; 1.36\ldots} \right),$
$\left. \left\{ a \to \boxed{⊘ \; 0.645\ldots}, \; g \to \boxed{⊘ \; 0.738\ldots} \right\} \right\}$

In[19]:= `N[%]`

Out[19]=

`{1.58451, {a → 0.644584, g → 0.737781}}`


\end{document}

%% file: img/intuition1.pdf_t
\begin{picture}(0,0)%
\includegraphics{intuition1.pdf}%
\end{picture}%
\setlength{\unitlength}{3947sp}%
\begingroup\makeatletter\ifx\SetFigFont\undefined%
\gdef\SetFigFont#1#2#3#4#5{%
  \reset@font\fontsize{#1}{#2pt}%
  \fontfamily{#3}\fontseries{#4}\fontshape{#5}%
  \selectfont}%
\fi\endgroup%
\begin{picture}(4169,699)(5004,-3730)
\put(5101,-3511){\makebox(0,0)[lb]{\smash{{\SetFigFont{12}{14.4}{\familydefault}{\mddefault}{\updefault}{\color[rgb]{0,0,0}$(0,M)$}%
}}}}
\put(5926,-3661){\makebox(0,0)[lb]{\smash{{\SetFigFont{12}{14.4}{\familydefault}{\mddefault}{\updefault}{\color[rgb]{0,0,0}$D^*(1,2)=M$}%
}}}}
\put(7801,-3511){\makebox(0,0)[lb]{\smash{{\SetFigFont{12}{14.4}{\familydefault}{\mddefault}{\updefault}{\color[rgb]{0,0,0}$(M-\epsilon,M+\epsilon)$}%
}}}}
\put(5926,-3361){\makebox(0,0)[lb]{\smash{{\SetFigFont{12}{14.4}{\familydefault}{\mddefault}{\updefault}{\color[rgb]{0,0,0}$D(1,2)=2M-\epsilon$}%
}}}}
\put(8311,-3184){\makebox(0,0)[lb]{\smash{{\SetFigFont{12}{14.4}{\familydefault}{\mddefault}{\updefault}{\color[rgb]{0,0,0}$2$}%
}}}}
\put(5317,-3178){\makebox(0,0)[lb]{\smash{{\SetFigFont{12}{14.4}{\familydefault}{\mddefault}{\updefault}{\color[rgb]{0,0,0}$1$}%
}}}}
\end{picture}%

%% file: img/intuition2.pdf_t
\begin{picture}(0,0)%
\includegraphics{intuition2.pdf}%
\end{picture}%
\setlength{\unitlength}{3947sp}%
\begingroup\makeatletter\ifx\SetFigFont\undefined%
\gdef\SetFigFont#1#2#3#4#5{%
  \reset@font\fontsize{#1}{#2pt}%
  \fontfamily{#3}\fontseries{#4}\fontshape{#5}%
  \selectfont}%
\fi\endgroup%
\begin{picture}(4907,2344)(4635,-5398)
\put(5251,-4561){\makebox(0,0)[lb]{\smash{{\SetFigFont{12}{14.4}{\familydefault}{\mddefault}{\updefault}{\color[rgb]{0,0,0}$\vdots$}%
}}}}
\put(8401,-4561){\makebox(0,0)[lb]{\smash{{\SetFigFont{12}{14.4}{\familydefault}{\mddefault}{\updefault}{\color[rgb]{0,0,0}$\vdots$}%
}}}}
\put(5101,-3511){\makebox(0,0)[lb]{\smash{{\SetFigFont{12}{14.4}{\familydefault}{\mddefault}{\updefault}{\color[rgb]{0,0,0}$(0,M)$}%
}}}}
\put(7801,-3511){\makebox(0,0)[lb]{\smash{{\SetFigFont{12}{14.4}{\familydefault}{\mddefault}{\updefault}{\color[rgb]{0,0,0}$(M-\epsilon,M+\epsilon)$}%
}}}}
\put(5101,-4111){\makebox(0,0)[lb]{\smash{{\SetFigFont{12}{14.4}{\familydefault}{\mddefault}{\updefault}{\color[rgb]{0,0,0}$(0,M)$}%
}}}}
\put(5101,-5086){\makebox(0,0)[lb]{\smash{{\SetFigFont{12}{14.4}{\familydefault}{\mddefault}{\updefault}{\color[rgb]{0,0,0}$(0,M)$}%
}}}}
\put(7801,-4111){\makebox(0,0)[lb]{\smash{{\SetFigFont{12}{14.4}{\familydefault}{\mddefault}{\updefault}{\color[rgb]{0,0,0}$(M-\epsilon,M+\epsilon)$}%
}}}}
\put(7801,-5086){\makebox(0,0)[lb]{\smash{{\SetFigFont{12}{14.4}{\familydefault}{\mddefault}{\updefault}{\color[rgb]{0,0,0}$(M-\epsilon,M+\epsilon)$}%
}}}}
\put(5707,-3854){\makebox(0,0)[lb]{\smash{{\SetFigFont{12}{14.4}{\familydefault}{\mddefault}{\updefault}{\color[rgb]{0,0,0}$2$}%
}}}}
\put(9174,-5334){\makebox(0,0)[lb]{\smash{{\SetFigFont{12}{14.4}{\familydefault}{\mddefault}{\updefault}{\color[rgb]{0,0,0}$2k$}%
}}}}
\put(5674,-3227){\makebox(0,0)[lb]{\smash{{\SetFigFont{12}{14.4}{\familydefault}{\mddefault}{\updefault}{\color[rgb]{0,0,0}$1$}%
}}}}
\put(5707,-5280){\makebox(0,0)[lb]{\smash{{\SetFigFont{12}{14.4}{\familydefault}{\mddefault}{\updefault}{\color[rgb]{0,0,0}$k$}%
}}}}
\put(9114,-3201){\makebox(0,0)[lb]{\smash{{\SetFigFont{12}{14.4}{\familydefault}{\mddefault}{\updefault}{\color[rgb]{0,0,0}$k+1$}%
}}}}
\end{picture}%

%% file: img/intuition3.pdf_t
\begin{picture}(0,0)%
\includegraphics{intuition3.pdf}%
\end{picture}%
\setlength{\unitlength}{3947sp}%
\begingroup\makeatletter\ifx\SetFigFont\undefined%
\gdef\SetFigFont#1#2#3#4#5{%
  \reset@font\fontsize{#1}{#2pt}%
  \fontfamily{#3}\fontseries{#4}\fontshape{#5}%
  \selectfont}%
\fi\endgroup%
\begin{picture}(3931,1528)(4785,-7425)
\put(8476,-7336){\makebox(0,0)[lb]{\smash{{\SetFigFont{12}{14.4}{\familydefault}{\mddefault}{\updefault}{\color[rgb]{0,0,0}$j$}%
}}}}
\end{picture}%

%% file: img/intuition4.pdf_t
\begin{picture}(0,0)%
\includegraphics{intuition4.pdf}%
\end{picture}%
\setlength{\unitlength}{3947sp}%
\begingroup\makeatletter\ifx\SetFigFont\undefined%
\gdef\SetFigFont#1#2#3#4#5{%
  \reset@font\fontsize{#1}{#2pt}%
  \fontfamily{#3}\fontseries{#4}\fontshape{#5}%
  \selectfont}%
\fi\endgroup%
\begin{picture}(3963,1528)(4785,-7425)
\put(8476,-6286){\makebox(0,0)[lb]{\smash{{\SetFigFont{12}{14.4}{\familydefault}{\mddefault}{\updefault}{\color[rgb]{0,0,0}$j'$}%
}}}}
\put(8476,-7336){\makebox(0,0)[lb]{\smash{{\SetFigFont{12}{14.4}{\familydefault}{\mddefault}{\updefault}{\color[rgb]{0,0,0}$j$}%
}}}}
\end{picture}%

%% file: img/invariant.pdf_t
\begin{picture}(0,0)%
\includegraphics{invariant.pdf}%
\end{picture}%
\setlength{\unitlength}{3947sp}%
\begingroup\makeatletter\ifx\SetFigFont\undefined%
\gdef\SetFigFont#1#2#3#4#5{%
  \reset@font\fontsize{#1}{#2pt}%
  \fontfamily{#3}\fontseries{#4}\fontshape{#5}%
  \selectfont}%
\fi\endgroup%
\begin{picture}(6968,2744)(2835,-7433)
\put(8476,-7336){\makebox(0,0)[lb]{\smash{{\SetFigFont{12}{14.4}{\familydefault}{\mddefault}{\updefault}{\color[rgb]{0,0,0}$k$}%
}}}}
\put(9023,-6207){\makebox(0,0)[lb]{\smash{{\SetFigFont{12}{14.4}{\familydefault}{\mddefault}{\updefault}{\color[rgb]{0,0,0}$P_{\ge j}(k)$}%
}}}}
\put(3720,-6727){\makebox(0,0)[lb]{\smash{{\SetFigFont{12}{14.4}{\familydefault}{\mddefault}{\updefault}{\color[rgb]{0,0,0}$N^-_{\ge j}(k)$}%
}}}}
\put(2850,-6207){\makebox(0,0)[lb]{\smash{{\SetFigFont{12}{14.4}{\familydefault}{\mddefault}{\updefault}{\color[rgb]{0,0,0}$N^-(k)$}%
}}}}
\put(8490,-5972){\makebox(0,0)[lb]{\smash{{\SetFigFont{12}{14.4}{\familydefault}{\mddefault}{\updefault}{\color[rgb]{0,0,0}$r$}%
}}}}
\put(8490,-6577){\makebox(0,0)[lb]{\smash{{\SetFigFont{12}{14.4}{\familydefault}{\mddefault}{\updefault}{\color[rgb]{0,0,0}$s$}%
}}}}
\put(9786,-5911){\makebox(0,0)[lb]{\smash{{\SetFigFont{12}{14.4}{\familydefault}{\mddefault}{\updefault}{\color[rgb]{0,0,0}$P(k)$}%
}}}}
\put(5188,-6128){\makebox(0,0)[lb]{\smash{{\SetFigFont{12}{14.4}{\familydefault}{\mddefault}{\updefault}{\color[rgb]{0,0,0}$j$}%
}}}}
\end{picture}%